\pdfoutput=1

\documentclass[acmsmall,screen,nonacm]{acmart}
\usepackage[disable]{todonotes}

\acmJournal{PACMPL}
\acmVolume{1}
\acmNumber{POPL} %
\acmArticle{1}
\acmYear{2027}
\acmMonth{1}
\acmDOI{} %
\startPage{1}

\setcopyright{none}

\usepackage{booktabs}   %
\usepackage{subcaption} %
\usepackage{import}     %

\usepackage[utf8]{inputenc}
\usepackage[english]{babel}
\usepackage{xparse}
\usepackage{xspace}
\usepackage{thmtools}
\usepackage{xcolor}
\usepackage{mathpartir}
\usepackage[]{microtype}
\usepackage{multirow}
\usepackage{makecell}
\usepackage{marvosym}
\usepackage{wasysym}
\usepackage{pifont}
\usepackage{mathtools}
\usepackage{stmaryrd}
\usepackage{scalerel}
\usepackage[nomath,notext]{stix} %
\usepackage{tensor}
\usepackage{xifthen}
\usepackage{iris}
\usepackage{prob}
\usepackage{pftools}
\usepackage{soul}
\usepackage{tikz}
\usepackage{threeparttable}
\usepackage{enumitem}
\usepackage{csquotes}
\usepackage{heaplang}
\usepackage[nameinlink]{cleveref}
\usepackage{bussproofs}
\usepackage[normalem]{ulem}
\usepackage{xfrac}
\usepackage{wrapfig}
\usetikzlibrary{calc, shapes, arrows, automata, patterns}

\makeatletter
\addto\extrasenglish{%
  \renewcommand*\chapterautorefname{\S\@gobble}
  \renewcommand*\sectionautorefname{\S\@gobble}
  \renewcommand*\subsectionautorefname{\S\@gobble}
  }
\makeatother
\newcommand*{\myeqref}[2][]{%
  \hyperref[{#2}]{#1(\ref*{#2})}%
}

\newcommand*{\lineref}[1]{\hyperref[#1]{line~\ref*{#1}}}
\newcommand*{\linerangeref}[2]{\hyperref[#1]{lines~\ref*{#1}}\hyperref[#2]{-\ref*{#2}}}
\newcommand{\mypageref}[1]{\hyperref[#1]{page~\pageref*{#1}}}

\crefname{section}{\S\!\!}{\S\!\!}
\crefname{subsection}{\S\!\!}{\S\!\!}

\DeclareUnicodeCharacter{2264}{$\le$}
\DeclareUnicodeCharacter{2208}{$\in$}
\DeclareUnicodeCharacter{2203}{$\exists$}
\DeclareUnicodeCharacter{25C1}{$\triangleleft$}
\DeclareUnicodeCharacter{2097}{${}_l$}
\DeclareUnicodeCharacter{2260}{$\neq$}
\DeclareUnicodeCharacter{2205}{$\emptyset$}
\DeclareUnicodeCharacter{222A}{$\cup$}
\DeclareUnicodeCharacter{228E}{$\uplus$}
\DeclareUnicodeCharacter{2200}{$\forall$}
\DeclareUnicodeCharacter{25A1}{$\always$}
\DeclareUnicodeCharacter{03B3}{$\gamma$}

\makeatletter
\providecommand*{\Dashv}{%
  \mathrel{%
    \mathpalette\@Dashv\vDash
  }%
}
\newcommand*{\@Dashv}[2]{%
  \reflectbox{$\m@th#1#2$}%
}
\makeatother
\makeatletter %
\def\arcr{\@arraycr}
\makeatother

\newcommand\ie{\emph{i.e.},\xspace}
\newcommand\eg{\emph{e.g.},\xspace}

\newcommand\wrt{\emph{w.r.t.}\xspace}

\makeatletter
\def\@parfont{\bfseries\itshape}
\makeatother

\newcommand{\mathleftalign}[1]{#1\hspace*{\linewidth-\widthof{\ensuremath{#1}}}}

\DeclareFontEncoding{LS1}{}{}
\DeclareFontSubstitution{LS1}{stix}{m}{n}
\DeclareSymbolFont{arrows2}{LS1}{stixsf}{m}{it}
\DeclareMathSymbol{\toea}{\mathrel}{arrows2}{"C1}

\newcommand{\phytr}{\vv{t}}
\newcommand{\anntr}{\vv{t}\!'}
\newcommand{\proph}[3][]{\ifblank{#1}{\textlog{proph}\!\left(#2,#3\right)}{\textlog{proph}^{#1}\!\left(#2,#3\right)}}
\newcommand{\prophtrace}[4][]{\ifblank{#1}{\textlog{prophTrace}\!\left(#2,#3,#4\right)}{\textlog{prophTrace}_{#1}\!\left(#2,#3,#4\right)}}
\newcommand{\ghostToken}[1]{\textlog{token}\left(#1\right)}

\newcommand{\ownstate}{\textlog S_{\!\authfrag}}
\newcommand{\cfgsafename}{\textlog{cfg-safe}}
\newcommand{\safename}{\textlog{safe}}

\newcommand{\noforkname}{\textlog{seq}}

\newcommand{\linname}{\textlog{lin}}
\newcommand{\trlinname}{\textlog{trlin}}
\newcommand{\atomspecname}{\textlog{atom-spec}}
\newcommand{\tracespecname}{\textlog{trace-spec}}
\newcommand{\mgcspecname}{\textlog{mgc-spec}}
\newcommand{\traceIsname}{\textlog{traceIs}}
\newcommand{\traceInvname}{\textlog{traceInv}}
\def\GFP #1.{\text{gfp}~#1.\spac}
\def\LFP #1.{\text{lfp}~#1.\spac}
\renewcommand{\emptyset}{\varnothing}
\renewcommand{\lbag}{\{\!|}
\renewcommand{\rbag}{|\!\}}
\newcommand{\ltrans}{\raisebox{0.3ex}{\scalebox{0.5}{\ensuremath{\ll}}}}
\newcommand{\rtrans}{\raisebox{0.3ex}{\scalebox{0.5}{\ensuremath{\gg}}}}
\newcommand{\IsP}{\textlog{IsP}}
\newcommand{\PState}{\textlog{PState}}
\newcommand{\MgcInv}{\textlog{MgcInv}}
\newcommand{\AU}{\textlog{AU}}
\newcommand{\lin}[3][]{\ifblank{#1}{\linname(#2,#3)}{\linname_{#1}(#2,#3)}}
\newcommand{\trlin}[2][]{\ifblank{#1}{\trlinname(#2)}{\trlinname_{#1}(#2)}}
\newcommand{\atomspec}[5][]{\ifblank{#1}{\atomspecname(#2,#3,#4,#5)}{\atomspecname_{#1}(#2,#3,#4,#5)}}
\newcommand{\tracespec}[5][]{\ifblank{#1}{\tracespecname(#2,#3,#4,#5)}{\tracespecname_{#1}(#2,#3,#4,#5)}}
\newcommand{\mgcspec}[5][]{\ifblank{#1}{\mgcspecname(#2,#3,#4,#5)}{\mgcspecname_{#1}(#2,#3,#4,#5)}}
\newcommand{\traceIs}[2]{\traceIsname(#1,#2)}
\newcommand{\traceInv}[2]{\traceInvname(#1,#2)}

\newcommand{\heaplang}{\textsf{HeapLang}}

\newcommand{\purestep}[1][]{\xrightarrow{{#1}}_{\textlog{pure}}}
\newcommand{\basestep}[1][]{\xrightarrow{{#1}}_{\textlog{base}}}
\newcommand{\advstep}[2][]{\xrightarrow[\smash{\raisebox{2pt}{$\scriptstyle{#1}$}}]{\smash{\raisebox{-2pt}{$\scriptstyle{#2}$}}}\mathrel{\vphantom{\to}_{\textlog{adv}}}}
\newcommand{\primstep}[1][]{\xrightarrow{{#1}}_{\textlog{prim}}}

\newcommand{\advsteps}[2][]{\xrightarrow[\smash{\raisebox{2pt}{$\scriptstyle{#1}$}}]{\smash{\raisebox{-2pt}{$\scriptstyle{#2}$}}}\mathrel{\vphantom{\to}^{*}_{\textlog{adv}}}}

\newcommand{\mconfigsteps}[2][]{\xrightarrow[\smash{\raisebox{2pt}{$\scriptstyle{#1}$}}]{\smash{\raisebox{-2pt}{$\scriptstyle{#2}$}}}\mathrel{\vphantom{\to}^{*}_{\textlog{mc}}}}

\newcommand{\cfgsafe}[3]{\cfgsafename(#1,#2,#3)}
\newcommand{\safe}[3][]{\ifblank{#1}{\safename(#2,#3)}{\safename_{#1}(#2,#3)}}

\newcommand{\nofork}[2]{\noforkname(#1,#2)}

\newcommand{\Conf}{\mathit{Cfg}}
\newcommand{\bool}{\mathbb{B}}

\newcommand{\prefixof}{\sqsubseteq_{\textlog{pre}}}

\newcommand{\suffixof}{\sqsubseteq_{\textlog{suf}}}

\newcommand{\mathraisebox}[2]{\mathchoice{\raisebox{#1}{$\displaystyle #2$}}{\raisebox{#1}{$#2$}}{\raisebox{#1}{$\scriptstyle #2$}}{\raisebox{#1}{$\scriptscriptstyle #2$}}}

\newcommand{\pure}[1]{\mathraisebox{0.55ex}{\left\ulcorner\kern-0.2em \mathraisebox{-0.55ex}{#1} \kern-0.2em\right\urcorner}}

\definecolor[named]{GrayOutColor}{cmyk}{0,0.0,0.0,0.6}

\definecolor[named]{ACMBlue}{cmyk}{1,0.1,0,0.1}
\definecolor[named]{ACMYellow}{cmyk}{0,0.16,1,0}
\definecolor[named]{ACMOrange}{cmyk}{0,0.42,1,0.01}
\definecolor[named]{ACMRed}{cmyk}{0,0.90,0.86,0}
\definecolor[named]{ACMLightBlue}{cmyk}{0.49,0.01,0,0}
\definecolor[named]{ACMGreen}{cmyk}{0.20,0,1,0.19}
\definecolor[named]{ACMPurple}{cmyk}{0.55,1,0,0.15}
\definecolor[named]{ACMDarkBlue}{cmyk}{1,0.58,0,0.21}

\colorlet{KeywordColor}{ACMDarkBlue}
\colorlet{EffectTagColor}{ACMOrange!60!black}
\colorlet{CommentColor}{black!65!white}
\colorlet{FunctionColor}{KeywordColor}
\colorlet{GhostColor}{ACMRed}
\colorlet{GhostKeywordColor}{GhostColor!75!KeywordColor}

\newcommand{\uq}{\textquotesingle}
\renewcommand{\langkw}[1]{{\text{\textcolor{KeywordColor}{\upshape\textlang{{\bfseries #1}}}}}}

\newcommand{\function}[1]{\text{\textcolor{FunctionColor}{\texttt{#1}}}}
\newcommand{\ghostcode}[1]{\begingroup\colorlet{KeywordColor}{GhostKeywordColor}\color{GhostColor}#1\endgroup}
\newcommand{\Some}{\operatorname{\langkw{some}}}
\newcommand{\None}{\operatorname{\langkw{none}}}

\allowdisplaybreaks

\begin{document}
\newtheorem{remark}[theorem]{Remark}
\renewcommand{\theHtheorem}{\thetheorem}
\renewcommand{\theHlemma}{\thetheorem}
\renewcommand{\theHcorollary}{\thetheorem}
\renewcommand{\theHproposition}{\thetheorem}
\renewcommand{\theHconjecture}{\thetheorem}
\renewcommand{\theHexample}{\thetheorem}
\renewcommand{\theHdefinition}{\thetheorem}
\renewcommand{\theHremark}{\thetheorem}
\let\logatomoldcaption\caption
\renewcommand{\caption}[1]{\logatomoldcaption{#1.}}

\title{Completeness of Logical Atomicity for Linearizability in Concurrent Separation Logic}

\author{Zichen Zhang}
\email{zichenzhang@nyu.edu}
\orcid{0009-0004-1151-6149}
\affiliation{%
  \institution{New York University}
  \city{New York}
  \country{USA}
}
\author{Simon Oddershede Gregersen}
\email{gregersen@cispa.de}
\orcid{0000-0001-6045-5232}
\affiliation{%
  \institution{CISPA Helmholtz Center for Information Security}
  \city{Saarbr{\"u}cken}
  \country{Germany}
}
\author{Joseph Tassarotti}
\email{jt4767@nyu.edu}
\orcid{0000-0001-5692-3347}
\affiliation{%
  \institution{New York University}
  \city{New York}
  \country{USA}
}
\authornote{Also affiliated with Amazon Web Services. This paper does not reflect the views of Amazon Web Services.}

\begin{abstract}

Linearizability is a standard correctness condition for concurrent data structures.
It guarantees that operations behave as if they took effect at some atomic instant between their call and return points.
Despite the central role linearizability plays, %
in the context of concurrent separation logics, prior work has argued for instead using a style of specification that internalizes the atomicity of operations in terms of the logic's reasoning rules, known as \emph{logical atomicity}.
These logically atomic specifications are intended to be easier to compose inside of the logic than linearizability.
Prior work has shown that in the Iris separation logic framework, a certain form of logically atomic specifications implies that a data structure is linearizable, which can be understood as a soundness theorem.
However, the converse remained an open question: for every linearizable data structure, is it always possible to derive a corresponding logically atomic specification?

This paper resolves this question in the affirmative.
We prove a completeness theorem for Iris that derives a logically atomic specification for any linearizable data structure.
As a consequence, we are able to embed a variety of linearizability proof techniques into Iris and use them to derive logically atomic specifications.
We apply this to three linearizability proof methods: aspect-oriented linearizability proofs, forward simulations with commit points, and meta-configuration tracking.
Using these embeddings, we derive logically atomic specifications for the Herlihy-Wing queue and the Baskets Queue.
We furthermore establish a connection between logical atomicity and an encoding of refinement in Iris that has been used in prior logical relations models.
This result allows us to transport logically atomic specifications across refinements, which we apply to the Folly MPMC queue implementation.
All of the results in this paper have been mechanized in the Rocq Prover.

\end{abstract}



\received{2026-07-09}

\maketitle

\section{Introduction}\label{sec:introduction}

Linearizability is a standard correctness condition for concurrent data structures~\citep{herlihy1990-linearizability,herlihy1987-linearizability,sela2021-linearizability-typo}.
It ensures that concurrent operations behave as if they took effect atomically in a way that is compatible with a sequential specification of the data structure.
Because linearizability is such a widely used correctness condition for concurrent data structures, a vast array of formal methods techniques related to linearizability have been developed.
These include various proof techniques~\citep{adadi1991-hpv,DBLP:journals/entcs/ColvinDG05,bouajjani2017-forward,jayanti2024-metaconfig,DBLP:conf/wdag/FeldmanE0RS18,henzinger2013-aspect,DBLP:conf/podc/OHearnRVYY10,DBLP:journals/pacmpl/FeldmanKE0NRS20}, program logics~\citep{DBLP:phd/ethos/Vafeiadis08, DBLP:conf/pldi/LiangF13,DBLP:conf/popl/TuronW11,hatti2026-lhl}, automated verifiers~\citep{vafeiadis2010-cave,zhu2015-poling,meyer2022-plankton,meyer2023-nekton}, and tools for checking for linearizability violations~\citep{DBLP:conf/pldi/BurckhardtDMT10,DBLP:conf/forte/HornK15a,DBLP:conf/cav/KovalFSTA23,DBLP:conf/cav/EmmiE19,DBLP:conf/cav/CirisciEFM20,DBLP:conf/ppopp/OzkanMN19,DBLP:journals/pacmpl/Han025}.
Alongside this, researchers have developed various adaptations of linearizability, including relaxations for weak memory~\citep{DBLP:conf/esop/BurckhardtGMY12,DBLP:conf/forte/DerrickS17,DBLP:conf/sefm/DohertyD16}, %
more compositional formulations~\citep{DBLP:journals/pacmpl/ValeSC23},
restrictions that guarantee stronger properties~\citep{DBLP:conf/stoc/GolabHW11,DBLP:conf/wdag/DenysyukW15}, and variants for durable or persistent memory~\citep{DBLP:conf/wdag/IzraelevitzMS16,DBLP:conf/podc/AttiyaBH18,aguilera2003-strict-lin}.

\paragraph{Logical Atomicity}
Yet despite the central role linearizability plays, a number of researchers have argued that linearizability, as formally defined, is not so easy to use to verify \emph{clients} of a linearizable data structure.
In particular, while concurrent separation logic (CSL) has been successfully used to verify a range of concurrent programs, and there are even ways to use CSL to \emph{prove} that a data structure is linearizable, once such a specification has been shown, it is difficult to use in pre/post-condition style reasoning about clients.
Thus, if one's goal is to use CSL to modularly verify a large application, linearizability may not be the most helpful specification for internal data structures and libraries used in the application.
As a result, prior work has developed alternative specifications that capture the atomicity of a data structure's operations internally in a program logic in a way that is easier to use when verifying clients.
Various forms of such \emph{logically atomic} specifications have been developed~\citep{da-rocha-pinto2014-tada,svendsen2013-hocap,iCAP,Iris1,nanevski2014-auth-monoid,DBLP:conf/popl/JacobsP11}.
The common idea is that when an operation is logically atomic, a client can reason about that operation using the same reasoning patterns that are supported for the primitive, physically atomic commands of the language.
For example, one of the key properties of physically atomic steps in most variants of CSLs is that they support \emph{invariant} rules that allow a thread to assume that an invariant assertion $P$ holds before the step, and require $P$ to be reestablished after the step.
Logically atomic operations are similarly allowed to access invariants, as long as the invariant can be shown to be preserved after the atomic operation takes effect.
This enables modular and compositional reasoning, in which one logically atomic data structure can be used in the implementation of another.

There were long understood to be informal analogies between these logically atomic specifications and linearizability.
For example, establishing logical atomicity typically requires identifying the atomic point where an operation's effect becomes visible, which is similar to proof techniques for linearizability that involve identifying the \emph{linearization point}.
Data structures whose linearizability is challenging to establish because of future-dependent linearization points similarly require special techniques like prophecy variables~\cite{adadi1991-hpv,jung2019-proph} for establishing logical atomicity.
\citet{birkedal2021-free} made this connection formal in one direction, by showing that if a data structure satisfies a logically atomic specification in Iris, a widely-used CSL framework~\cite{Iris1,IrisGroundUp}, then the data structure is linearizable.

But what about the opposite direction: if a data structure is linearizable, is it always possible to prove a logically atomic specification for it?
We can view this as a kind of completeness property for logical atomicity, and view the previous result of \citeauthor{birkedal2021-free} as a soundness theorem.
No prior work has established a completeness theorem for logical atomicity, but such a result would be useful for several reasons.
First, it would tell us that the logic is not too ``weak'' or ``missing'' any rules or reasoning techniques: if a data structure is really linearizable, then at least in principle, a logically atomic specification can be shown using the logic.
Second, it means that if we can prove, external to the logic, that a data structure is linearizable, then we can automatically reflect that into the logic as a logically atomic specification for the data structure.
This is valuable because, as mentioned above, a range of proof techniques and automated verification methods have been developed for establishing linearizability, and many of these have no existing analogue inside Iris or other program logics.
Instead of having to extend Iris and its formulation of logical atomicity to support these proof techniques, completeness for logical atomicity would allow us to just directly apply these proof techniques to establish linearizability and then obtain a logically atomic specification for free.

This paper proves the first completeness theorem for Iris's logically atomic triples.
For concreteness, we establish our result for \heaplang, the default language included with Iris, but our technique is generally applicable to other languages used with the Iris framework.
The key idea behind our proof is an inductive forward simulation-style argument.
To derive the logically atomic Hoare triple about an operation, we maintain an invariant tracking the trace of call/return operations.
From the assumption that the operation is linearizable, we know that there must exist some linearization point between the call/return operations, and we use these linearization points to know when to apply the ``atomic update'' assertion used in Iris's logically atomic triples.
The main challenge is that formally, when given a partial prefix of execution, the definition of linearizability does not necessarily ensure that the linearization points we extract in this way are compatible with future extensions to the trace.
We solve this by instrumenting the methods of a concurrent data structure with \emph{prophecy variables} that record the arguments of each method call and the eventual return value.
Using these prophecy variables, the proof ``predicts'' the complete trace of calls and returns, so that we can extract a consistent set of linearization points from the linearizability assumption.

As alluded to above, by using this completeness theorem, we are able to prove linearizability of data structures using arbitrary proof methods and then derive logically atomic triples back in Iris.
We apply this concretely with three linearizability proof techniques that have no existing Iris analogue: the aspect-oriented proof~\citep{henzinger2013-aspect}, the forward-simulation with commit points technique developed by \citet{bouajjani2017-forward}, and meta-configuration tracking developed by \citet{jayanti2024-metaconfig}.
What is powerful about each of these methods is that they are capable of proving linearizability of data structures with future-dependent linearization points.
We have applied these methods to prove logically atomic triples of two data structures with such future dependency, the Herlihy-Wing queue~\citep{herlihy1990-linearizability} and the Baskets Queue~\citep{DBLP:conf/opodis/HoffmanSS07}.
While our completeness results show that one could in principle just use prophecy variables directly to verify these data structures (and indeed, prior work has done so for the Herlihy-Wing queue~\cite{jung2019-proph}), we argue that it is useful to be able to reuse off-the-shelf proofs and ideas that have been developed in prior work, instead of being forced to redevelop a proof in terms of prophecy variables.

\paragraph{Refinement}

In addition to linearizability and logical atomicity, there is a third commonly used approach to specifying concurrent data structures: refinement.
Using contextual refinement, one can capture that a data structure behaves in an atomic way by proving that it refines a coarse-grained implementation that wraps a sequential implementation with a global lock or atomic block primitive.
Because the contextual refinement shows that the data structure's observable behaviors are a subset of the coarse-grained version, clients can reason about the data structure \emph{as if} it were atomic.
Many formal verification techniques have been developed for relational reasoning about concurrent data structures and can be used for proving such refinements~\citep{DBLP:conf/cav/HawblitzelPQT15,DBLP:conf/cav/KraglQ18,DBLP:conf/popl/LiangFF12}, including several based on concurrent separation logics and Iris in particular~\citep{frumin2021-reloc,gaher2022simuliris,DBLP:conf/icfp/TuronDB13,DBLP:conf/popl/Krogh-Jespersen17}.

How does this kind of contextual refinement specification relate to linearizability or logical atomicity?
\citet{filipovic2010-ctx-refine} first proved an equivalence between linearizability and contextual refinement, and subsequent works have generalized this to a number of settings, including relaxed memory~\citep{DBLP:conf/vmcai/DongolJRA18,DBLP:journals/pacmpl/SinghL23} and durable variants of linearizability~\citep{DBLP:journals/pacmpl/ValeW0YS24}.

But for logical atomicity, there is no existing formal connection with contextual refinement.
Some prior work has observed that from a logical atomicity specification, it is often easy to establish a contextual refinement to an atomic implementation of a data structure~\citep{frumin2021-reloc}.
However, there has previously been no way to go from a proof of a refinement to a logically atomic triple.

The second main contribution of this paper is a result connecting logical atomicity to refinement.
In particular, if $e_I$ refines $e_A$, and $e_A$ satisfies a logically atomic specification, then our result enables us to derive a logically atomic specification for $e_I$.
Among other things, this is useful because it allows one to establish logical atomicity incrementally through a chain of refinements.
Rather than having to prove $e_I$ is logically atomic all at once, we can introduce a series of intermediary programs $e_1, \dots, e_n$, and prove $e_{I} \precsim e_{1}$, $e_1 \precsim e_2$, $e_2 \precsim e_3$, \dots, $e_n \precsim e_A$, deduce a logically atomic specification for $e_A$ and then transport this transitively to a specification for $e_I$.
This pattern of establishing chains of refinements is common in the literature based on refinement style reasoning, and our results make this accessible for the first time as a proof method for logical atomicity as well.
For example, \citet{DBLP:conf/cpp/VindumFB22} have proved that the MPMC queue from Meta's C++ library Folly \cite{folly} refines a coarse-grained queue with a global lock, but their result is not accessible from a unary separation logic.
With the help of our technique, we are finally able to get a logically atomic triple for the Folly MPMC queue on top of their result.

\begin{figure}
\centering
\begin{tikzpicture}
\node [draw,rounded corners] at (-4.5,0) (refine) {Refinement}; %

\node [draw,rounded corners] at (4.5,0) (la) {Logical Atomicity}; %

\node [draw,rounded corners] at (0,2) (lin) {Linearizability}; %

\draw[transform canvas={xshift=0.15cm},-Stealth] (la) -- node[above,sloped] {\citet{birkedal2021-free}} (lin);
\draw[transform canvas={xshift=-0.15cm},-Stealth] (lin) -- node[left=6pt] {\cref{sec:logatom}} (la);
\draw[Stealth-Stealth] (refine) -- node[above,sloped] {\citet{filipovic2010-ctx-refine}} (lin);
\draw[-Stealth] (refine) -- node[above] {\cref{sec:refine}} (la);
\end{tikzpicture}
\caption{Relation between Linearizability and Its Alternatives}
\label{fig:lin-triangle}
\end{figure}
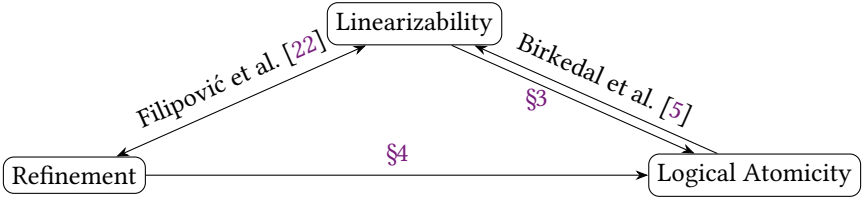

\paragraph{Contributions}
\Cref{fig:lin-triangle} diagrammatically shows our theoretical results and where they fit with the existing literature.
To summarize, our work makes the following contributions:
\begin{itemize}
\item The first completeness result for logical atomicity showing that any linearizable data structure satisfies a logically atomic triple specification.
\item A proof connecting refinement to logical atomicity, which in particular allows atomic triples to be transported across a refinement.
\item Machine-checked reformulations of various linearizability proof techniques in the context of Iris, and examples applying them to several challenging data structures.
\end{itemize}
Our work is mechanized in the Iris separation logic framework~\cite{IrisGroundUp} and the Rocq Prover.

\section{Preliminaries}\label{sec:prelim}

This section gives preliminary background material needed for our results.
We start with the semantics of the concurrent language we consider.
Next, we give a formal definition of linearizability as used in our proof.
Finally, we describe the basic features of the Iris separation logic that our completeness proof uses.

\subsection{Language}
\begin{figure}
\begin{align*} \val,\valB\in\Val\bnfdef{}&\TT\mid z\mid b\mid\loc\mid\Rec\lvarF(\lvar)= \expr\mid(\val,\valB)\mid\Inl(\val)\mid\Inr(\val)\mid\cdots\tag{\ensuremath{z \in \integer,b\in\bool, \loc \in \Loc}}\\
\expr\in\Expr\bnfdef{}&\val\mid\var\mid\Rec\lvarF(\lvar)= \expr\mid\expr_1(\expr_2)\mid\HLOp_1\expr\mid \expr_1\HLOp_2\expr_2\mid\If \expr then \expr_1 \Else \expr_2
  \mid
  (\expr_1,\expr_2) \\
  \mid{}&
  \Fst(\expr) \mid
  \Snd(\expr)
  \mid{}
  \Inl(\expr) \mid
  \Inr(\expr) \mid
  (\Match \expr with \Inl => \expr_1 | \Inr => \expr_2 end)\\
  \mid{}&\Alloc \expr\mid\deref\expr\mid\expr_1\gets\expr_2\mid\CAS(\expr_0,\expr_1,\expr_2)\mid\Fork\expr\mid\cdots\\
\lctx\in\Lctx\bnfdef{}&\bullet\mid\expr(\lctx)\mid\lctx(\val)\mid\expr\HLOp_2\lctx\mid\lctx\HLOp_2\val\mid\cdots\tag{Standard right-to-left evaluation context}\\
\pstate\in\State\eqdef{}&\Loc\fpfn\Val\hspace{4em}\Conf\eqdef\Expr^*\times\State
\end{align*}
\begin{mathpar}
{\mprset{fraction={\ \ \ }}
\inferhref{BaseBeta}{rule:base-beta}
{((\Rec\lvarF(\lvar)= \expr)(\val),\pstate)\basestep}{(\subst {\subst \expr \lvarF {(\Rec\lvarF(\lvar)= \expr)}} \lvar \val,\pstate,\nil)}}
\and\cdots
\and
\inferhref{BaseAlloc}{rule:base-alloc}
{\loc\notin\dom(\pstate)}{(\Alloc \val,\pstate)\basestep(\loc,\mapinsert{\loc}{\val}{\pstate},\nil)}
\and
\inferhref{BaseFork}{rule:base-fork}
{}{(\Fork\expr,\pstate)\basestep(\TT,\pstate,[\expr])}
\and
\cdots
\end{mathpar}
\caption{Syntax and Semantics for \heaplang{} (without Prophecy Variables)}\label{fig:heaplang}
\end{figure}

Although our completeness proof method generalizes, for concreteness, our results are proven with respect to the default \heaplang{} included in the Iris Rocq development.
This is a mini ML-like language with general references and concurrency primitives.
Its syntax and semantics are shown in \Cref{fig:heaplang}.
Constructions in this figure are standard.
The language has syntactic categories for expressions $\expr$, values $\val$, states $\pstate$, and evaluation contexts $\lctx$.
The operational semantics is first specified by a base step relation $(\expr_1,\pstate_1)\basestep(\expr_2,\pstate_2,\vv*{\expr}f)$, which says that under state $\pstate_1$, expression $\expr_1$ reduces to $\expr_2$, updates the state to $\pstate_2$, and spawns threads $\vv*ef$.
A pure step $\expr_1\purestep\expr_2$ is a special base step that is deterministic and independent of the state.
The base step relation is then closed under evaluation contexts $\lctx$ to get the primitive step relation $\primstep$, which is lifted to a concurrent thread pool step relation $\tpstep$, as expressed in the following rules:
\begin{mathpar}
\inferhref{BasePrim}{rule:base-prim}
{(\expr_1,\pstate_1)\basestep(\expr_2,\pstate_2,\vv*{\expr}f)}
{(K[\expr_1],\pstate_1)\primstep(K[\expr_2],\pstate_2,\vv*{\expr}f)}
\and
\inferhref{PrimTp}{rule:prim-tp}
{(\expr_1,\pstate_1)\primstep(\expr_2,\pstate_2,\vv*{\expr}f)}
{(\vv* \expr a\dplus \expr_1::\vv*\expr b,\pstate_1)\tpstep(\vv* \expr a\dplus \expr_2::\vv*\expr b\dplus\vv*\expr f,\pstate_2)}
\end{mathpar}
In particular, in \ruleref{rule:prim-tp}, the threads of the program are represented as a list of expressions, from which one expression is non-deterministically selected to step.
In the resulting program configuration, that thread is updated along with the state, and the new list $\vv*{\expr}f$ of forked threads (if any) is appended to the end of the thread list.
We denote the reflexive transitive closure of $\tpstep$ as $\tpsteps$.

A standard correctness condition for a \emph{closed} program is safety or partial correctness.
It means executing a program $\expr$ from initial state $\pstate$ never gets stuck, and if $\expr$ terminates with a value $\val$ and state $\pstate'$, then $\phi(\val,\pstate')$ holds.
\begin{definition}[Safety]\label{def:safe}
An expression $\expr$ in state $\pstate$ is safe \wrt $\phi$, denoted as $\safe[\phi]{\expr}{\pstate}$, if
\begin{enumerate}
\item $\All \val, \vv\expr\!',\pstate'.([\expr],\pstate)\tpsteps(\val::\vv\expr\!',\pstate')\Ra \phi(\val,\pstate')$, and
\item $\All \vv\expr\!',\pstate'.([\expr],\pstate)\tpsteps(\vv\expr\!',\pstate')\Ra\All \expr'\in\vv\expr\!'.\expr'\in\Val\lor\red(\expr',\pstate')$,
\end{enumerate}
where $\red(\expr,\pstate)\eqdef\Exists \expr',\pstate',\vv*\expr f.(\expr,\pstate)\primstep(\expr',\pstate',\vv*\expr f)$.
When $\phi$ is $(\Lam \val,\pstate.\TRUE)$, we omit $\phi$ and write $\safe{\expr}{\pstate}$.
\end{definition}

\subsection{Linearizability}
For a \emph{partial} program, \eg a library, in addition to safety, we also want to ensure the interaction between the library and other parts of the program is correct.
Particularly for concurrent data structures, linearizability~\citep{herlihy1987-linearizability, herlihy1990-linearizability, sela2021-linearizability-typo} is a standard correctness condition that relates a concurrent object's behaviors to a sequential specification.
There are multiple equivalent definitions of linearizability.
In this paper, we use a version that is stated in terms of \emph{linearization points}%
, following the definition used by \citet{birkedal2021-free}.
At a high level, this definition requires that the execution history $H$ of a concurrent object can be annotated with marked \emph{linearization points}, which are instantaneous moments at which an operation's effect appears to take place.
By ordering operations according to the order of their linearization points, we obtain a total order on operations.
Then, the observed call/return behaviors of the execution must match those of a sequential execution that runs operations in this total order.

To state this definition formally, we will assume that the interface for a data structure supports two functions, $\function{init}$, which creates a new object, and $\function{op}$, which takes an object $\textit{obj}$ and an argument $x$ and invokes an operation on the object with request $x$.
Restricting to a single $\function{op}$ function does not result in a loss of generality, because multiple operations can be emulated by including a tag as part of the argument that selects the intended operation.

\paragraph{Sequential Behavior}
The sequential specification of an object is given in terms of an abstract type $S$ of object states.
This state is initialized to $s_0 \in S$ upon the creation of an object.
The effects of operations are specified in terms of relations on this state.

\begin{definition}[Sequential Behavior]
A sequential behavior $\mu$ of data structure $(\function{init},\function{op})$ is a triple of $(s_0,\textit{vr},\textit{rr})\colon S\times (\Val\to\mProp)\times (S\times\Val\times S\times \Val\to\mProp)$, such that
\begin{itemize}
\item The initial state of an object is $s_0$;
\item For object $\textit{obj}$ of state $s$ and request $x$, if $\textit{vr}(x)$, then performing $\function{op}(\textit{obj},x)$ returns $y$ and updates the object to state $s'$ such that $\textit{rr}(s,x,s',y)$.
\end{itemize}
\end{definition}

The valid request predicate $\textit{vr}$\/ restricts which arguments to $\function{op}$ are considered well-formed or valid.
Given a valid argument, the behavior of $\function{op}$ is captured by a request-response relation $\textit{rr}$.
This relation may be non-deterministic.
In case there are multiple sequential behaviors in question, we write $\mu.\textit{vr}$ or $\mu.\textit{rr\/}$ to disambiguate.

\begin{example}[Bag]\label{example:bag}
A bag, \ie a multiset, is a container that supports $\function{push}$ and $\function{pop}$ operations.
The $\function{push}$ operation adds a new element to the container, while
$\function{pop}$ removes an arbitrary element from the container and returns this element, or reports that the container is empty.
The sequential behavior of a bag, $\mu_{\textit{bag}}$, is defined as
\begin{align*}
s_0\eqdef{}&\varnothing\\
\textit{vr}(x)\eqdef{}&(\Exists v.x=\Some(v))\lor x=\None\\
\textit{rr}(s,x,s',y)\eqdef{}&(\Exists v.x=\Some(v)\land s'=\lbag v\rbag\uplus s \land {y=()} )\\
&\lor ({x=\None} \land {s=\varnothing} \land {s'=\varnothing} \land {y = \None}) \\
&\lor ({x=\None} \land {\Exists v. v\in s} \land {s'=s\setminus\lbag v\rbag} \land {y = \Some(v)})
\end{align*}
where $\None$ and $\Some(\val)$  are syntactic sugar for $\Inl(\TT)$ and $\Inr(\val)$, respectively.
This definition encodes $\function{push}(b,v)$ as $\function{op}(b,\Some(v))$ and $\function{pop}(b)$ as $\function{op}(b,\None)$.
For the return value of a pop, $\None$ means the bag is empty, and $\Some v$ means element $v$ is removed from the bag.
\end{example}

\paragraph{Linearizability of a Trace}

We first define what it means for a single interaction of a data structure with a client to be linearizable.
An interaction is represented by a call-return trace, and a linearization of such a trace is obtained by adding suitable linearization point events to the trace.

Formally, we define $\Sigma$ to be the type of events that can occur in physical traces, and let $\Sigma_{\textlog{lin}}$ be an extension that includes linearization events.
\begin{align*}
\Sigma&\eqdef\{(\tau,\textlog{Call}(v))\mid\tau\in\textdom{Tag},v\in\Val\}\cup\{(\tau,\textlog{Ret}(v'))\mid\tau\in\textdom{Tag},v'\in\Val\}\\
\Sigma_{\textlog{lin}}&\eqdef\Sigma\cup\{(\tau,\textlog{Lin}(v,v'))\mid\tau\in\textdom{Tag},v,v'\in\Val\}
\end{align*}
Here, $\tau$ is a tag that relates the call, linearization, and return event of an operation.
The parameters of call and return events represent the argument to the operation and the returned value, respectively.
Linearization events take two parameters that match the argument and return value for the call/return that they correspond to.
A \emph{physical trace} $\phytr \in \Sigma^*$ is then a finite sequence of call and return events, while an \emph{annotated trace} $\anntr \in \Sigma_{\textlog{lin}}^*$ is a finite sequence that may also include linearization events.
Both types of traces include the empty sequence $\nil$.
We define various projection operations on traces that return the subtrace consisting of certain selected events.
Given a trace $\vv u$, $\pi_\tau(\vv u)$ is the subtrace of events with tag $\tau$, $\pi_\Sigma({\vv u})$ is the subtrace of physical events, and $\pi_{\textlog{lin}}(\vv u)$ is the subtrace of linearization events.
A trace is well-formed if the call, linearization, and return events of each tag occur in order, and each event only occurs once.
\begin{align}
\textlog{well-formed}(\phytr)&\eqdef \All \tau.\Exists v,v'.\pi_\tau(\phytr)\prefixof[(\tau,\textlog{Call}(v));(\tau,\textlog{Ret}(v'))]\label{eq:well-formed}\\
\textlog{well-formed}_\textlog{lin}(\anntr)&\eqdef \All \tau.\Exists v,v'.\pi_\tau(\anntr)\prefixof[(\tau,\textlog{Call}(v));(\tau,\textlog{Lin}(v,v'));(\tau,\textlog{Ret}(v'))]\label{eq:wll-formed}
\end{align}
where $\vv{x}\prefixof\vv{y}$ means $\vv{x}$ is a prefix of $\vv{y}$.

\begin{definition}[Trace Linearizability]\label{def:trlin}
A trace $\phytr\in \Sigma^*$ is linearizable \wrt a sequential behavior $\mu$, written $\trlin[\mu]{\phytr}$, if the parameter of every call event in $\phytr$ satisfies $\textit{vr}$ and there exists an annotated trace $\anntr$ and state $s'$ such that $\phytr=\pi_\Sigma(\anntr)$, $\textlog{well-formed}_\textlog{lin}(\anntr)$, and $\textlog{sound}_\mu(\pi_{\textlog{lin}}(\anntr),s_0,s')$, where
\begin{align*}
\textlog{sound}_\mu(\nil,s,s')&\eqdef s=s'\\
\textlog{sound}_\mu(\anntr\dplus[(\tau,\textlog{Lin}(x,y))],s,s'')&\eqdef\Exists s'.\textlog{sound}_\mu(\anntr,s,s')\land\textit{vr}(x)\land \textit{rr}(s',x,s'',y)
\end{align*}
\end{definition}
Here, the annotated trace $\smash{\anntr}$ is a linearization of $\phytr$ because $\smash{\anntr}$ consists of the same physical events as $\phytr$, and the linearization event of each operation occurs between its call event and return event.
Moreover, $\textlog{sound}_\mu(\pi_{\textlog{lin}}(\smash{\anntr}),s_0,s')$ says starting from the initial state $s_0$, processing operations according to the linearization events in the order found in $\smash{\anntr}$ will transform the object to state $s'$, \ie the linearization matches the sequential behavior $\mu$.

In \Cref{def:trlin}, a ``\textlog{sound}'' non-empty trace is obtained by appending a new linearization event to the end of another sound trace.
Alternatively, one can also choose to prepend an event to the beginning of another sound trace.

\begin{remark}\label{remark:sound-cons}
$\textlog{sound}_\mu((\tau,\textlog{Lin}(x,y))::\anntr,s,s'')\Lra\Exists s'.\textit{vr}(x)\land\textit{rr}(s,x,s',y)\land\textlog{sound}_\mu(\anntr,s',s'')$
\end{remark}

\paragraph{Linearizability of a Data Structure}
With trace linearizability defined, we now define linearizability of a data structure by requiring that the traces generated by interactions with the data structure are linearizable.
However, we must restrict our attention to \emph{valid} traces, in which all arguments to $\function{op}$ are valid.
In addition, the client must only interact with objects via calls to $\function{op}$.
For example, if an object has internal private state, the client is not allowed to break the data structure's abstractions and directly modify this private state.
Aside from these restrictions, however, a client can send any possible valid sequence of requests to the object and spawn an unbounded number of threads to send requests to the object concurrently.
To represent this set of traces, we define a new reduction relation that models these interaction rules.
This relation is defined over \emph{configurations} that track the operations that are executing.

\begin{definition}[Configuration]
Given an object $\textit{obj}$ of a data structure $(\function{init},\function{op})$,
a configuration $\rho$ is a triple of $(\vv{\pi},\vv{\expr},\pstate)\colon(\textdom{Tag}\times\Expr^?)^*\times\Expr^*\times\State$, where
\begin{itemize}
\item $\vv{\pi}$ is a list representing the executing operations on the object.
Each element of the list is a pair of tag $\tau$ and expression $\expr$.
For operations that have finished execution and returned, the expression component of the pair is a special value $\bot$.
\item $\vv{e}$ is a list of other threads executing in the system that were created by the executing operations in $\vv{\pi}$.
\item $\pstate$ is the state of the program, \eg the heap. The state of the data structure is part of $\pstate$.
\end{itemize}
\end{definition}

Then, a step of interaction between an object and a client is modeled by a \emph{most adversarial step} relation between configurations.

\begin{definition}[Most Adversarial Step]
Given an object \textit{obj}, method $\function{op}$, and predicate $\textit{vr}$, a \emph{most adversarial step} $\rho\mathrel{\smash{\advstep[\textit{obj}.\function{op}/\textit{vr}]{\phytr}}}\rho'$ is defined by the following set of inference rules:

{\arraycolsep=1.4pt%
$$\begin{array}[t]{rlll}
\multicolumn{4}{l}{\textlabel{rule:adv-call}{\textsc{Adv-Call}}}\\
  (\vv{\pi},\vv{\expr},\pstate)
  &\advstep[\textit{obj}.\function{op}/\textit{vr}]{\mathmakebox[3.3em]{[(\tau,\textlog{Call}(v))]}}
  &(\vv{\pi}\dplus[(\tau,\function{op}(\textit{obj},v))],\vv{\expr},\pstate)
  &~\text{if}~ \tau\notin\vv{\pi}.1 \text{ and } \textit{vr}(v) \\
\multicolumn{4}{l}{\textlabel{rule:adv-interm}{\textsc{Adv-Interm}}}\\
  (\vv*\pi1\dplus (\tau,\expr_1) ::\vv*\pi2,\vv{\expr},\pstate_1)
  &\advstep[\textit{obj}.\function{op}/\textit{vr}]{\mathmakebox[3.3em]{\nil}}
  &(\vv*\pi1\dplus(\tau,\expr_2)::\vv*\pi2,\vv{\expr}\dplus\vv*{\expr}{f},\pstate_2)
  &~\text{if}~ (\expr_1,\pstate_1)\primstep(\expr_2,\pstate_2,\vv*{\expr}{f}) \\
\multicolumn{4}{l}{\textlabel{rule:adv-child}{\textsc{Adv-Child}}}\\
  (\vv \pi,\vv*{\expr}{a}\dplus \expr_1::\vv*{\expr}{b},\pstate_1)
  &\advstep[\textit{obj}.\function{op}/\textit{vr}]{\mathmakebox[3.3em]{\nil}}
  &(\vv \pi,\vv*{\expr}{a}\dplus\expr_2::\vv*{\expr}{b}\dplus\vv*{\expr}f,\pstate_2)
  &~\text{if}~ (\expr_1,\pstate_1)\primstep(\expr_2,\pstate_2,\vv*{\expr}{f}) \\
\multicolumn{4}{l}{\textlabel{rule:adv-ret}{\textsc{Adv-Ret}}}\\
  (\vv*\pi1\dplus(\tau,\val)::\vv*\pi2,\vv \expr,\pstate)
  &\advstep[\textit{obj}.\function{op}/\textit{vr}]{\mathmakebox[3.3em]{[(\tau,\textlog{Ret}(v))]}}
  &(\vv*\pi1\dplus(\tau,\bot)::\vv*\pi2,\vv\expr,\pstate)
  &~\text{if}~ \val\in\Val %
\end{array}$$}

\ruleref{rule:adv-call} represents the invocation of a new operation with argument $v$, which must be valid according to the $\textit{vr}$ predicate.
The transition is labeled with a call event with a fresh tag $\tau$, and an expression executing $\function{op}(\textit{obj},v)$ is added to the list $\vv{\pi}$ with the same tag $\tau$.
\ruleref{rule:adv-interm} performs a step of execution of an operation $\expr$ in the list $\vv{\pi}$, adding any forked child threads to the $\vv{\expr}$ component of the configuration.
\ruleref{rule:adv-child} executes a step of reduction of one of these child threads.
Finally, \ruleref{rule:adv-ret} selects an operation that has terminated with a value $\val$ and labels the transition with a return event, replacing the value $\val$ with $\bot$ in the configuration afterwards.

The non-deterministic nature of these rules represents the ability of a client to invoke an arbitrary number of concurrent operations, so long as the arguments are valid, in the sense that they satisfy $\textit{vr}$.
These invocations are interleaved with steps of the operations themselves, as well as any children they produce.
Critically, the non-deterministic steps by the environment do not alter the local state of the data structure.

The relation $\smash{\advsteps{}}$ is the reflexive transitive closure of $\smash{\advstep{}}$, which concatenates the labels on transitions.
\begin{mathpar}
\inferhref{Advs-Nil}{rule:advs-nil}{}
{\raisebox{-1ex}{\ensuremath{\rho\advsteps[\textit{obj}.\function{op}/\textit{vr}]{\nil}\rho}}}\and
\inferhref{Advs-Cons}{rule:advs-cons}{}
{\raisebox{-1ex}{\ensuremath{\rho_0\advstep[\textit{obj}.\function{op}/\textit{vr}]{\vv*t1}\rho_1\land\rho_1\advsteps[\textit{obj}.\function{op}/\textit{vr}]{\vv*t2}\rho_2\implies
 \rho_0\advsteps[\textit{obj}.\function{op}/\textit{vr}]{\vv*t1\dplus\vv*t2}\rho_2}}}
\end{mathpar}
We omit $\textit{obj}$, $\function{op}$, and $\textit{vr}$, when they are clear from the context, and simply write $\rho_1\advsteps{\phytr}\rho_2$.
\end{definition}

Using the most adversarial step relation, we can now define linearizability of a data structure.
\begin{definition}[Linearizability]\label{def:linear}
A data structure $(\function{init},\function{op})$ is linearizable \wrt a sequential behavior $\mu$, denoted as $\lin[\mu]{\function{init}}{\function{op}}$, if
\begin{enumerate}
\item $\safe{\function{init}()}{\varnothing}$, \ie the evaluation of $\function{init}()$ never gets stuck; and
\item if $([\function{init}()],\varnothing)\tpsteps([\textit{obj}],\pstate)$ and $(\nil,\nil,\pstate)\mathrel{\smash{\advsteps[\textit{obj}.\function{op}/\textit{vr}]{\phytr}}}(\vv{\pi}, \vv{e}, \pstate')$, then all expressions in $\vv{\pi}$ and $\smash{\vv{\expr}}$ are not stuck, and $\smash{\trlin[\mu]{\phytr}}$.
\end{enumerate}
\end{definition}
The second part of this definition runs $\function{init}$ under the normal reduction relation $\tpsteps$ until it terminates in some value $\textit{obj}$, which is then used as part of the most adversarial step reduction.
Every trace $\phytr$ generated in this manner must be linearizable.
In addition, we require that execution of $\function{init}$ and $\function{op}$ never gets stuck, in order to rule out implementations that are unsafe and trigger some form of undefined behavior.

\subsection{An Iris-Based Separation Logic}
\begin{figure}
\begin{align*}
P,Q\in\Prop\bnfdef{}&\pure{\varphi}\mid P\land Q\mid P\lor Q\mid\All x. P\mid\Exists x. P
\mid\loc\mapsto\val
\tag{\ensuremath{\varphi} is a meta-level proposition}\\
{}\mid{}& P\ast Q\mid P\wand Q
\mid\later P
\mid\hoare{P}{\expr}{\Ret \val.Q}
\mid\pvs P
\mid\ownGhost{\gamma}{\authfull M}
\mid k\hookrightarrow^\gamma v
\mid\cdots
\end{align*}
\begin{mathpar}
\inferhref{HtConseq}{rule:wp-wand}
{\hoare{P}{\expr}{\Phi}\\P'\proves P\\\All x.(\Phi(x)\proves\Phi'(x))}
{\proves\hoare{P'}{\expr}{\Phi'}}
\mprset{fraction=--\ast}
\and
\inferhref{HtFrame}{rule:wp-frame}
{\hoare{P}{\expr}{\Ret \val. Q}}{\hoare{P\ast R}{\expr}{\Ret\val.Q\ast R}}
\and
\inferhref{HtValue}{rule:wp-value}
{\Phi(\val)}{\hoare{P}{\val}{\Phi}}
\and
\inferhref{HtSeq}{rule:wp-seq}
{\hoare{P}{\expr}{\Ret v. Q}\\\All v.\hoare{Q}{\lctx[\val]}{\Phi}}
{\hoare{P}{\lctx[\expr]}{\Phi}}
\and
{\mprset{defaultfraction}\inferhref{HtPure}{rule:wp-pure}
{\expr_1\purestep\expr_2}
{\hoare{P}{\expr_2}{\Phi}\proves\hoare{\later P}{\expr_1}{\Phi}}}
\and
\inferhref{HtAlloc}{rule:wp-alloc}
{}{\hoare{\TRUE}{\Alloc\val}{\Ret\loc.\loc\mapsto\val}}
\and
\inferhref{HtLoad}{rule:wp-load}
{}{\hoare{\loc\mapsto\val}{\deref\loc}{\Ret\var. \var = \val \ast \loc\mapsto\val}}
\and
\inferhref{HtStore}{rule:wp-store}
{}{\hoare{\loc\mapsto\val}{\loc\gets\valB}{\Ret \var. \var = \TT \ast \loc\mapsto\valB}}
\and
\inferhref{HtFupd}{rule:wp-fupd}
{\hoare{P}{\expr}{\Phi}}{\hoare{\pvs P}{\expr}{\Phi}}
\and
\mprset{defaultfraction}
\inferhref{MapAlloc}{rule:map-alloc}
{}{\proves \Exists \gamma.\pvs\ownGhost{\gamma}{\authfull\varnothing}}
\and
\inferhref{MapLookup}{rule:map-lookup}
{}{\ownGhost{\gamma}{\authfull M}\ast k\hookrightarrow^\gamma v
\proves\pure{M(k)=v}}
\and
\inferhref{MapInsert}{rule:map-insert}
{k\notin\dom(M)}
{\ownGhost{\gamma}{\authfull M}
\proves\pvs\ownGhost{\gamma}{\authfull(\mapinsert{k}{v}{M})}
\ast k\hookrightarrow^\gamma v}
\and
\inferhref{MapUpdate}{rule:map-update}
{}{\ownGhost{\gamma}{\authfull M}\ast k\hookrightarrow^\gamma v
\proves\pvs\ownGhost{\gamma}{\authfull(\mapinsert{k}{w}{M})}
\ast k\hookrightarrow^\gamma w}
\and
\inferhref{MapDelete}{rule:map-delete}
{}{\ownGhost{\gamma}{\authfull M}\ast k\hookrightarrow^\gamma v
\proves\pvs\ownGhost{\gamma}{\authfull(M\setminus\set{k})}}
\and
\inferhref{HigherOrderMapLookup}{rule:higher-order-map-lookup}
{}{\ownGhost{\gamma}{\authfull M}\ast k\hookrightarrow^\gamma P\proves\Exists P'.\pure{M(k)=P'}\ast\later (P\equiv P')}
\end{mathpar}
\caption{A Separation Logic for \heaplang{}}\label{fig:heaplang-sl}
\end{figure}

\Cref{fig:heaplang-sl} shows a subset of rules from the default separation logic for \heaplang{} provided by Iris.%
\footnote{As usual for Iris, the actual logic we work on uses another connective $\wpre{e}{\Phi}$ and the Hoare triple is syntactic sugar for a certain form of $\wprename$.}
This logic can be used to derive judgments
of the shape $P \proves Q$, where $P$ and $Q$ are separation logic assertions.
Separation logic enriches
propositional logic with the separating conjunction $\ast$ and separating implication $\wand$, sometimes called the ``magic wand''.
The separating conjunction $P \ast Q$ expresses ownership of $P$ and $Q$ separately, \ie $P$ and $Q$ hold for disjoint parts of the heap.
In particular, $P \proves P \ast P$ does not always hold.
Iris is an \emph{affine} separation logic that admits the weakening rule $P\ast Q\proves P$.
Assertion $\pure{\varphi}$ embeds a meta-level proposition (\ie a Rocq \texttt{Prop}) into the separation logic.
The points-to assertion $\loc \mapsto \val$ says that location $\loc$ stores the value $\val$.
We omit most of the rules for proving entailments; see Chapter 3 of \citet{iris-lecture-notes} for an expository introduction to separation logic and Iris.

To reason about programs, we use the Hoare triple $\hoare{P}{\expr}{\Ret\val.Q}$.
Here, $\val$ is a binder for the value returned by executing the expression $\expr$.
We omit this binder if $\val$ does not appear free in $Q$.
When writing inference rules, we use two different styles.
When the line dividing the premises from the conclusion is a standard horizontal bar, then this rule should be read as an implication in the meta-logic.
When the horizontal bar is replaced by a $\wand$, we use a rule with premises $P_{1}, \ldots, P_{n}$ and conclusion $Q$ as notation for the separation logic entailment $(P_1 \ast \cdots \ast P_n) \vdash Q$.

Most of the Hoare triple rules in \Cref{fig:heaplang-sl} are standard (ignore the later modality ``$\later$'' for now).
\ruleref{rule:wp-wand} is the rule of consequence from Hoare logic. %
\ruleref{rule:wp-frame} allows for local reasoning about the subset of state an expression accesses.
\ruleref{rule:wp-seq} generalizes the standard sequential composition rule of Hoare logic to an arbitrary evaluation context $K$.
The rules \ruleref{rule:wp-alloc}, \ruleref{rule:wp-load}, \ruleref{rule:wp-store} reason about operations on mutable state by updating the points-to assertion appropriately.
Finally, \ruleref{rule:wp-pure} allows for reasoning about pure reduction steps of the program, which do not require any state assertions.

\paragraph{Ghost State and Updates}
To reason about the physical state of a program, it is often useful to maintain some additional \emph{auxiliary} or \emph{ghost} states to track some abstract aspects of the program~\cite{OwickiG76}.
These ghost states can be updated by the prover using certain rules.

While Iris supports a large variety of ghost state \cite{IrisGroundUp}, this paper mainly uses (higher-order) ghost maps.
A ghost map comes with two types of resources: a full or authoritative resource $\ownGhost{\gamma}{\authfull M}$ carrying the knowledge of the whole map and a fragmentary resource $k\hookrightarrow^\gamma v$ carrying the value $v$ at key $k$.
The fragmentary resource is also known as \emph{ghost points-to}.
Here, $\gamma$ is the name of the ghost map.
Rule \ruleref{rule:map-alloc} allocates an empty ghost map with some fresh name.
In this rule, ``$\pvs$'' is called the update modality, meaning that something is provable by updating the underlying ghost state.
The update modality can be eliminated using \ruleref{rule:wp-fupd}.
There are two lookup rules for ghost maps: if the ghost map stores first-order values like $\mathbb{N}$ or $\Expr$, then when owning both $\ownGhost{\gamma}{\authfull M}$ and $k\hookrightarrow^\gamma v$, one immediately learns that $M(k)$ is equal to $v$ by \ruleref{rule:map-lookup}; but if the ghost map stores higher-order data like $\Prop$, one can only use \ruleref{rule:higher-order-map-lookup}.
This rule still concludes that there must exist some (higher-order) value $P'$ at $M(k)$, but $P'$ is only equivalent to $P$ under a later modality ``$\later$''.
Intuitively, $\later P$ means $P$ does not hold immediately, but will hold after the program executes one step.
This also explains why there is a later modality in the precondition of \ruleref{rule:wp-pure}---the later modality captures the fact that the program has taken one step.

\paragraph{Soundness and Completeness}

The \heaplang{} program logic is sound and complete.
\begin{theorem}[Soundness~\cite{IrisGroundUp} and Completeness~\cite{hostert2026-completeness}]\label{thm:sound-complete}
$$\proves\hoare{\TRUE}{\expr}{\Ret\val.\pure{\varphi(\val)}}\iff\safe[\Lam\val,\_.\varphi(\val)]{\expr}{\varnothing}$$
\end{theorem}

Note that \Cref{thm:sound-complete} only deals with closed programs.
It does not necessarily say anything about our ability to prove specifications for individual operations or functions of a library, or whether we can prove specifications that are useful for clients.
In contrast, specifications like linearizability are about a data structure that is intended to be used as part of a larger client application.
As we will see in the next section, our results instead provide a completeness theorem that has to do with reusable specifications about a library's operations.

\section{Additional Details about the Proof of \Cref{thm:logatom}}

This section presents additional technical details about the proof of \Cref{thm:logatom} outlined in \Cref{sec:logatom}.

\subsection{The $\textlog{prophTrace}$ Predicate}
\label{sec:app:typed-proph}
Assertion $\prophtrace[\textit{obj}]{\rho}{\prophid}{\phytr}$ is defined as
\begin{align*}
\prophtrace[\textit{obj}]{\rho}{\prophid}{\phytr}&\eqdef\Exists\vv\val.\proph{\prophid}{\vv\val}\ast\textlog{decValidTr}(\textit{obj},\rho,\phytr,\textit{decTr}(\vv{\val}))\\
\textit{decTr}(\Inl(\tau,x)::\vv\val)&\eqdef(\tau,\textlog{Call}(x))::\textit{decTr}(\vv\val)\\
\textit{decTr}(\Inr(\tau,y)::\vv\val)&\eqdef(\tau,\textlog{Ret}(y))::\textit{decTr}(\vv\val)\\
\textit{decTr}(\_)&\eqdef\nil\\
\textlog{decValidTr}(\textit{obj},\rho,\phytr,\vv*{t}{\textit{raw}})&\eqdef\textlog{ValidTr}(\textit{obj},\rho,\phytr)\\
&\hspace*{-3ex}\land(\phytr=\vv*{t}{\textit{raw}}\lor\Exists n<|\vv*{t}{\textit{raw}}|.\phytr=\vv*{t}{\textit{raw}}[0..n-1]\land\lnot\textlog{ValidTr}(\textit{obj},\rho,\vv*{t}{\textit{raw}}[0..n]))\\
\textlog{ValidTr}(\textit{obj},\rho,\phytr)&\eqdef\Exists \rho'.\rho\advsteps[\textit{obj}]{\phytr}\rho'
\end{align*}
The idea is to decode the untyped trace into a typed trace and discard ill-typed traces.
Specifically, the raw resolution list of values is first decoded into a physical trace by function $\textit{decTr}$, and the trace is then further decoded to a ``valid'' trace by relation $\textlog{decValidTr}$.
Function $\textit{decTr}$ tries decoding the head of the raw resolution list, and if it fails, it returns an arbitrary dummy value, which is $\nil$ in this case.
Because the program will always resolve the prophecy variable to $\Inl$ or $\Inr$, this failure will never happen, but we still need to account for the failure case as we have not proved anything about future resolutions yet.
The idea for $\textlog{decValidTr}$ is similar.
However, because predicate $\textlog{ValidTr}$ is undecidable in general, we cannot write a function to check whether a trace is generated by a valid interaction and discard its invalid part.
This relation consists of two parts: $\textlog{ValidTr}(\textit{obj},\rho,\phytr)$ asserts the result of decoding is a valid trace.
The second part says, if $\vv*t{\textit{raw}}$ is fully valid before decoding, then, the whole raw trace will be kept; otherwise, $\vv*t{\textit{raw}}$ must start to become invalid at some point, \ie there exists an index $n$, such that the first $n$ elements in $\vv*t{\textit{raw}}$ form a valid trace, but appending one more element makes the trace invalid.

$\prophtrace[\textit{obj}]{\rho}{\prophid}{\phytr}$ satisfies the following rules:
\begin{mathparpagebreakable}
\inferhref{HtNewProphTrace}{rule:wp-newproph-trace}
{}{\hoare{\TRUE}{\ghostcode{\NewProph}}{\Ret \prophid.\All \textit{obj},\rho. \Exists \phytr.\prophtrace[\textit{obj}]{\rho}{\prophid}{\phytr}}}
\and
\inferhref{HtResolveTrace-Call}{rule:wp-resolve-trace-call}
{\tau\notin\vv\pi.1\\\textit{vr}(x)}
{\curlybracket{\prophtrace[\textit{obj}]{(\vv\pi,\vv\expr,\pstate)}{\prophid}{\vv*t1}}\\\\
{\ghostcode{\Resolve \prophid to \Inl(\tau,x)}}\\\\
\curlybracket{\Ret \TT.\Exists \vv*t2.\pure{\vv*t1=(\tau,\textlog{Call}(x))::\vv*t2} \ast \prophtrace[\textit{obj}]{(\vv\pi\dplus[(\tau,\function{op}(\textit{obj},x))],\vv\expr,\pstate)}{p}{\vv*t2}}}
\and
\inferhref{HtResolveTrace-Ret}{rule:wp-resolve-trace-ret}
{y\in\Val\\\textlog{NoDup}((\vv*\pi1\dplus(\tau,y)::\vv*\pi2).1)}
{\curlybracket{\prophtrace[\textit{obj}]{(\vv*\pi1\dplus(\tau,y)::\vv*\pi2,\vv\expr,\pstate)}{\prophid}{\vv*t1}}\\\\
{\ghostcode{\Resolve \prophid to \Inr(\tau,y)}}\\\\
\curlybracket{\Ret \TT.\Exists \vv*t2.\pure{\vv*t1=(\tau,\textlog{Ret}(y))::\vv*t2} \ast \prophtrace[\textit{obj}]{(\vv*\pi1\dplus(\tau,\bot)::\vv*\pi2,\vv\expr,\pstate)}{p}{\vv*t2}}}
\and
\inferhref{ResolveTrace-Interm}{rule:wp-resolve-trace-interm}
{(\expr_1,\pstate_1)\primstep(\expr_2,\pstate_2,\vv*\expr f)}
{\mprset{fraction=--\ast}
\inferrule{\prophtrace[\textit{obj}]{(\vv*\pi1\dplus(\tau,\expr_1)::\vv*\pi2,\vv\expr,\pstate_1)}{\prophid}{\vv*t1}}
{\Exists \vv*t2.\pure{\vv*t2\prefixof\vv*t1}\ast \prophtrace[\textit{obj}]{(\vv*\pi1\dplus(\tau,\expr_2)::\vv*\pi2,\vv\expr\dplus\vv*\expr f,\pstate_2)}{\prophid}{\vv*t2}}}
\and
\inferhref{ResolveTrace-Child}{rule:wp-resolve-trace-child}
{(\expr_1,\pstate_1)\primstep(\expr_2,\pstate_2,\vv*\expr f)}
{\mprset{fraction=--\ast}
\inferrule{\prophtrace[\textit{obj}]{(\vv\pi,\vv*{\expr}{a}\dplus\expr_1::\vv*\expr b,\pstate_1)}{\prophid}{\vv*t1}}
{\Exists \vv*t2.\pure{\vv*t2\prefixof\vv*t1}\ast \prophtrace[\textit{obj}]{(\vv\pi,\vv*\expr a\dplus\expr_2::\vv*\expr b\dplus\vv*\expr f,\pstate_2)}{\prophid}{\vv*t2}}}
\end{mathparpagebreakable}

Rule \ruleref{rule:wp-newproph-trace} creates a fresh typed prophecy variable $\prophid$.
This rule allows us to associate $\prophid$ with any object at any time after $\prophid$ was created.
A critical step to prove this rule is to show that for any trace $\vv*{t}{\textit{raw}}$, there exists a trace $\phytr$ such that $\textlog{decValidTr}(\textit{obj},\rho,\phytr,\vv*t{\textit{raw}})$.
To do so, we consider each prefix of $\vv*t{\textit{raw}}$: $\nil,\vv*t{\textit{raw}}[0..0],\vv*t{\textit{raw}}[0..1],\dots,\vv*t{\textit{raw}}$, and find the last prefix still satisfying $\textlog{ValidTr}$.
Notably, this requires using the law of excluded middle because given a program from a Turing-complete language, deciding whether it can generate a trace is equivalent to the halting problem~\cite{rice1953}.

Rules \ruleref{rule:wp-resolve-trace-call} and \ruleref{rule:wp-resolve-trace-ret} handle the emission of physical events.
To emit a $\textlog{Call}$/$\textlog{Ret}$ event, one must show it is valid to perform this event on the current configuration.
These rules then resolve the prophecy variable and meanwhile update the configuration.
In the premise of \ruleref{rule:wp-resolve-trace-ret}, \textlog{NoDup} means no duplication, \ie a list does not contain the same item twice.

Rules \ruleref{rule:wp-resolve-trace-interm} and \ruleref{rule:wp-resolve-trace-child} handle the evaluation of the configuration.
Since there is no event for intermediate steps, these rules only update the configuration but do not resolve the prophecy variable.
Note that performing an intermediate step could truncate the trace:
Relation \textlog{decValidTr} implies that $\vv*t1$ is the \emph{longest} possible trace emitted by evaluating the current configuration, which requires the current configuration to take a specific first step, but if the prim step used by the rules is not the expected first step, then the new configuration would not be able to emit a trace as long as $\vv*t1$.
For example, program ``$\If \langkw{NondetBool}() then \function{op}(\textit{obj},1) \Else ()$'' can emit a trace of $[(\tau,\textlog{Call}(1))]$ under all possible executions, but can emit $\nil$ if the first prim step is $(K[\langkw{NondetBool}],\pstate)\primstep(K[\False],\pstate,\nil)$.
This explains why $\vv*t2$ can only be a prefix of $\vv*t1$.

\subsection{The $I_{\textlog{atom}}$ Invariant}
\label{sec:app:invariant}

\subsubsection{The Exclusive Authoritative Ghost State}
The actual invariant does not use a ghost map to track the abstract sequential state of the object, but rather uses a more restrictive resource $\ownGhost{\gamma_s}{\authfull\exinj(s)}$.
For the definition of $\PState_\mu$, we actually use $\ownGhost{\gamma_s}{\authfrag\exinj(s)}$.
The exclusive authoritative ghost resource satisfies the rules below.
\begin{mathpar}
\inferhref{ExAuthAlloc}{rule:exauth-alloc}
{}{\proves\pvs\Exists \gamma.\ownGhost{\gamma}{\authfull \exinj(x)}\ast\ownGhost{\gamma}{\authfrag \exinj(x)}}
\and
\inferhref{ExAuthAgree}{rule:exauth-agree}
{}{\ownGhost{\gamma}{\authfull \exinj(x)}\ast\ownGhost{\gamma}{\authfrag \exinj(y)}\proves\pure{x=y}}
\and
\inferhref{ExAuthUpdate}{rule:exauth-update}
{}{\ownGhost{\gamma}{\authfull \exinj(x)}\ast\ownGhost{\gamma}{\authfrag \exinj(y)}\proves\pvs\ownGhost{\gamma}{\authfull \exinj(z)}\ast\ownGhost{\gamma}{\authfrag \exinj(z)}}
\end{mathpar}

\subsubsection{Definition and Properties of \textit{pendingOp}}
\label{sec:app:invariant-consistency}
To compute $\textit{pendingOp}(\anntr)$, let's consider a tag $\tau$ of some event that has happened or will happen.
If $(\tau,\textlog{Call}(\_))\in \anntr$, then operation $\tau$ has not been called yet, so it is definitely not pending.
If $(\tau,\textlog{Call}(\_))\notin \anntr$ but $(\tau,\textlog{Lin}(\_,\_))\in \anntr$ or $(\tau,\textlog{Ret}(\_))\in \anntr$, then operation $\tau$ must have been called before the current time, and has not yet returned, so it is definitely pending.
Finally, if no event in $\anntr$ has tag $\tau$, then $\tau$ may have returned before the current time, or $\tau$ may have been called, but will never return.
The latter is possible because we work on partial correctness, so an operation is permitted to loop infinitely.
As a result, in the last case, we cannot determine whether $\tau$ is pending or not, and $\textit{pendingOp}(\anntr)$ will conservatively exclude $\tau$ in its result.
This also explains why the invariant uses $\textit{pendingOp}(\anntr)\subseteq \dom(\Phi s)$ rather than strict equality: because function $\textit{pendingOp}$ is an under-approximation, some pending operations can be missing in $\textit{pendingOp}(\anntr)$.

\begin{remark}\label{remark:pendingop-init}
$\textlog{well-formed}_{\textlog{lin}}(\anntr)\implies\textit{pendingOp}(\anntr)=\varnothing$
\end{remark}
\begin{remark}\label{remark:pendingop-call}
$\textit{pendingOp}(\anntr)\subseteq \set{\tau}\cup\textit{pendingOp}((\tau,\textlog{Call}(x))::\anntr)$
\end{remark}
\begin{remark}\label{remark:pendingop-lin}
$\textit{pendingOp}(\anntr)\subseteq\textit{pendingOp}((\tau,\textlog{Lin}(x,y))::\anntr)$
\end{remark}
\begin{remark}\label{remark:pendingop-ret}
$\textit{pendingOp}(\anntr)\subseteq\textit{pendingOp}((\tau,\textlog{Ret}(y))::\anntr)\setminus\set{\tau}$
\end{remark}

\subsubsection{Properties of Partial Linearizability}
\label{sec:app:invariant-partial-lin}
The following properties of $\textlog{partial-well-formed}_{\textlog{lin}}$ and $\cfgsafename$ (defined in \Cref{sec:logatom-inv}) are used in the proof of \Cref{thm:logatom}.

\begin{remark}\label{remark:partial-wf} $\textlog{partial-well-formed}_{\textlog{lin}}(\anntr)\land \anntr'\suffixof\anntr\implies\textlog{partial-well-formed}_{\textlog{lin}}(\anntr')$
\end{remark}

\begin{remark}\label{remark:cfg-safe-init}
$\lin[\mu]{\function{init}}{\function{op}}\land([\function{init}()],\varnothing)\tpsteps([\textit{obj}],\pstate)\implies\cfgsafe{\nil}{\nil}{\pstate}$.
\end{remark}

\begin{remark}\label{remark:cfg-safe-step}
$\cfgsafename(\rho)\land\smash{\rho\advsteps{\phytr}\rho'}\implies\cfgsafename(\rho')$
\end{remark}

\subsection{Detailed Proof of \Cref{thm:logatom}}
\label{sec:app:logatom-proof}
Part of the proof is stated in terms of the $\wprename$ assertion. The connection between $\wprename$ and Hoare triple is
$$\hoare{P}e{\Ret v. Q}\eqdef\Square\left(\All\Phi. P\wand\later(\All v.Q\wand\Phi(v))\wand\wpre{e}{\Phi}\right)$$
The proof consists of four parts.

\subsubsection{Specification of Wrapped $\function{init}$}
We first prove specification (\ref{eq:init-spec}):
$$\hoare{\TRUE}{\textit{wrapi}(\function{init})()}{\Ret \textit{obj}. \Exists \gamma.\IsP_\mu(\gamma,\textit{obj})\ast\PState_\mu(\gamma,s_0)}$$

Observe that Clause (1) of \Cref{def:linear} can be strengthened to $\safe[\textit{vo}]{\function{init}()}{\varnothing}$ where $\textit{vo}(\val,\pstate)\eqdef ([\function{init}()],\varnothing)\tpsteps ([\val],\pstate)$.
Using Theorem 34 by \citet{hostert2026-completeness}, we obtain a specification for $\function{init}$:
$$\hoare{\TRUE}{\function{init}()}{\Ret \textit{obj}.\Exists \pstate.\ownstate(\pstate)\ast\pure{\textit{vo}(\textit{obj},\pstate)}}$$

Using this specification and \ruleref{rule:wp-newproph-trace}, and picking the initial configuration as $(\nil,\nil,\pstate)$, we obtain
$$\ownstate(\pstate)\ast\pure{([\function{init}()],\varnothing)\tpsteps([\textit{obj}],\pstate)}\ast\prophtrace[\textit{obj}]{(\nil,\nil,\pstate)}{\prophid}{\phytr}$$
and we are obligated to prove
$$\Exists\gamma.\IsP_\mu(\gamma,(\textit{obj},\prophid))\ast\PState_\mu(\gamma,s_0)$$
which is equivalent to
$$\Exists \gamma,\gamma_\pi,\gamma_e,\gamma_\phi,\gamma_s.\ownGhost{\gamma}{\Square(\gamma_\pi,\gamma_e,\gamma_\phi,\gamma_s)}\ast\ownGhost{\gamma_s}{\authfrag \exinj(s_0)}\ast I_{\textlog{atom}_\mu}(\gamma,\textit{obj},\prophid)$$

The ownership of the typed prophecy variable $\prophtrace[\textit{obj}]{(\nil,\nil,\pstate)}{\prophid}{\phytr}$ guarantees $(\nil,\nil,\pstate)\advsteps{\phytr} \rho'$ for some $\rho'$.
Further because $([\function{init}()],\varnothing)\tpsteps([\textit{obj}],\pstate)$, by $\lin[\mu]{\function{init}}{\function{op}}$, we have $\trlin[\mu]{\phytr}$.
Namely, there exists an annotated trace $\anntr$, such that $\phytr=\pi_\Sigma(\anntr)$, $\textlog{well-formed}_{\textlog{lin}}(\anntr)$, and $\Exists s'.\textlog{sound}_\mu(\pi_{\textlog{lin}}(\anntr),s_0,s')$.
By \Cref{remark:cfg-safe-init}, $\cfgsafe{\nil}{\nil}{\pstate}$ holds.

Allocate the initial ghost resources as such:
$$\ownGhost{\gamma_\pi}{\authfull\nil}\ast\ownGhost{\gamma_e}{\authfull\nil}\ast\ownGhost{\gamma_\phi}{\authfull\varnothing}\ast\ownGhost{\gamma_s}{\authfull \exinj(s_0)}\ast\ownGhost{\gamma_s}{\authfrag \exinj(s_0)}\ast\ownGhost{\gamma}{\Square(\gamma_\pi,\gamma_e,\gamma_\phi,\gamma_s)}$$

It is then straightforward to establish $I_{\textlog{atom}}$ for this initial state: the operation lists and receipt map are empty, so (\ref{fig:I-atom-part-au})--(\ref{fig:I-atom-part-dom}) hold vacuously; (\ref{fig:I-atom-part-lin}) follows from $\textlog{well-formed}_{\textlog{lin}}(\anntr)$ and the $\cfgsafename$ fact just established; and the allocated ghost resources together with $\textlog{prophTrace}$ prove (\ref{fig:I-atom-part-own}).

This concludes the proof of the specification for $\textit{wrapi}(\function{init})$.

\subsubsection{The Call of an Operation}\label{sec:logatom-proof-call}
We now turn to the proof of specification (\ref{eq:op-spec}):
$$\pure{\textit{vr}(x)}\ast\IsP_\mu(\gamma,\textit{obj})\vdash\ahoare{\PState_\mu(\gamma,s)}{\textit{wrapo}(\function{op})(\textit{objp},x)}{\Ret y.\Exists s'.\PState_\mu(\gamma,s')\ast\pure{\textit{rr}(s,x,s',y)}}$$

By the definition of logically atomic triples, this is equivalent to proving $\AU_{\cdots}(\Phi)\wand\wpre{\textit{wrapo}(\function{op})(\textit{objp},x)}{\Phi}$ for an arbitrary $\Phi$ under the hypotheses that $\textit{vr}(x)$ and $\knowInv{}{I_{\textlog{atom}_\mu}(\gamma,\textit{obj},p)}$.

We first use \ruleref{rule:wp-newghostid} to allocate a fresh token for this operation and obtain resource $\ghostToken{\tau}$.
We now need to verify $\ghostcode{\Resolve p to \Inl(\tau,x)}$.
This requires opening the invariant to register the new operation.

By repeatedly using \ruleref{rule:token-ne} for each component in (\ref{fig:I-atom-part-token}), we have $\tau\notin\vv\pi.1$.

Using \ruleref{rule:wp-resolve-trace-call}, we update $\prophtrace{\rho}{p}{\pi_\Sigma(\anntr)}$ to $\prophtrace{\rho'}{p}{\vv u}$, where $\rho'=(\vv\pi\dplus[(\tau,\function{op}(\textit{obj},x))],\vv\expr,\pstate)$ and $\pi_\Sigma(\anntr)=(\tau,\textlog{Call}(x))::\vv u$.
Split $\anntr$ into three parts: $\anntr=\phytr_{\textlog{lin}}\dplus (\tau,\textlog{Call}(x))::\anntr'$, such that $\vv u=\pi_\Sigma(\anntr')$.
Because $(\tau,\textlog{Call}(x))$ is the first physical event in $\anntr$, $\phytr_{\textlog{lin}}$ only consists of $\textlog{Lin}$ events.

Update ghost state in the invariant:
\begin{align*}
\ownGhost{\gamma_\pi}{\authfull\textit{l2m}(\textit{activeOp}(\vv\pi))}&\vsWI
\ownGhost{\gamma_\pi}{\authfull\textit{l2m}(\textit{activeOp}(\vv\pi\dplus[(\tau,\function{op}(\textit{obj},x))]))}\ast\tau\hookrightarrow^{\gamma_\pi} \function{op}(\textit{obj},x)\\
\ownGhost{\gamma_\phi}{\authfull \Phi s}&\vsWI
\ownGhost{\gamma_\phi}{\authfull (\mapinsert{\tau}{\Phi}{\Phi s})}\ast\tau\hookrightarrow^{\gamma_\phi}\Phi
\end{align*}
Note that $\ownGhost{\gamma_e}{\authfull\vv{\expr}}$ and $\ownGhost{\gamma_s}{\authfull \exinj(s)}$ stay unchanged.
The former is unchanged because we are only registering a new operation, which should not affect existing child threads.
For the latter, we will update the abstract state of the object later when linearizing operations in $\vv*t{\textlog{lin}}$.

Let's try to close the invariant and see what is left:
Most of (\ref{fig:I-atom-part-own}) can be proved by the new states, except for $\ownGhost{\gamma_s}{\authfull \exinj(s)}$.
Part (\ref{fig:I-atom-part-token}) holds by inserting the new token $\ghostToken{\tau}$.
Part (\ref{fig:I-atom-part-dom}) holds by \Cref{remark:pendingop-call}.
Part (\ref{fig:I-atom-part-lin}) is proved by \Cref{remark:partial-wf,remark:cfg-safe-step}.
The only part left is (\ref{fig:I-atom-part-au}), which requires
linearizing operations in $\vv*t{\textlog{lin}}$ using \Cref{lem:helping} below and inserting the $\AU_{\cdots}(\Phi)$ in the context into the big separation conjunction.

\begin{lemma}[Helping]\label{lem:helping}
Let $\vv*{t}{\textlog{lin}}$ be a trace of linearization events, $\anntr'$ be an annotated trace, and $\Phi s$ be a finite map from tags to $\Prop$'s, such that $\textlog{partial-well-formed}_{\textlog{lin}}(\vv*{t}{\textlog{lin}}\dplus\anntr')$, $\textlog{sound}_\mu(\pi_{\textlog{lin}}(\vv*{t}{\textlog{lin}}\dplus\anntr'),s,s_{\textit{end}})$, and $\vv*{t}{\textlog{lin}}.1\subseteq \dom(\Phi s)$, then
$$\mprset{fraction=--\ast}
\inferrule{
\raisebox{2ex}{\ensuremath{\ownGhost{\gamma}{\Square(\gamma_\pi,\gamma_e,\gamma_\phi,\gamma_s)}}}\\
\raisebox{2ex}{\ensuremath{\ownGhost{\gamma_s}{\authfull\exinj(s)}}}\\
\raisebox{2ex}{\ensuremath{\Sep_{\tau\gets \Phi\in\Phi s} \textit{interpret}(\vv*{t}{\textlog{lin}}\dplus\anntr',\tau,\Phi)}}}
{\pvs\Exists s'.\pure{\textlog{sound}_\mu(\pi_{\textlog{lin}}(\anntr'),s',s_{\textit{end}})}\ast
\ownGhost{\gamma_s}{\authfull \exinj(s')}\ast
\Sep_{\tau\gets \Phi\in\Phi s} \textit{interpret}(\anntr',\tau,\Phi)}$$
\end{lemma}
\begin{proof}
Induction on list $\vv*{t}{\textlog{lin}}$.
The base case is trivial: if $\vv*{t}{\textlog{lin}}=\nil$, then we don't need to update anything.

For the inductive case: suppose $\vv*{t}{\textlog{lin}}=(\tau,\textlog{Lin}(x,y))::\vv*{t}{\textlog{lin}}'$.
To linearize the first event $(\tau,\textlog{Lin}(x,y))$, we need to prove the goal below by firing the $\AU$
$$\mprset{fraction=--\ast}
\inferrule{\ownGhost{\gamma_s}{\authfull \exinj(s)}\\\AU_{\PState_\mu(\gamma,s),\Ret y.\Exists s'.\PState_\mu(\gamma,s')\ast\pure{\textit{rr}(s,x,s',y)}}(\Phi)}
{\pvs\Exists s'.\pure{\textlog{sound}_\mu(\pi_{\textlog{lin}}(\vv*t{\textlog{lin}}'\dplus\anntr'),s',s_{\textit{end}})}\ast
\ownGhost{\gamma_s}{\authfull\exinj(s')}\ast\Phi(y)}$$

To fire the $\AU$, we first obtain $\ownGhost{\gamma_s}{\authfrag\exinj(s)}$ from the precondition of the $\AU$.
Because $\textlog{sound}_\mu((\tau,\textlog{Lin}(x,y))::\pi_{\textlog{lin}}(\vv*t{\textlog{lin}}'\dplus\anntr'),s,s_{\textit{end}})$, by \Cref{remark:sound-cons}, there exists a new state $s'$ and a response $y$ such that $\textit{vr}(x)$, $\textit{rr}(s,x,s',y)$, and $\textlog{sound}_\mu(\pi_{\textlog{lin}}(\vv*t{\textlog{lin}}'\dplus\anntr'),s',s_{\textit{end}})$.

Update the ghost state: $\ownGhost{\gamma_s}{\authfull \exinj(s)}\ast\ownGhost{\gamma_s}{\authfrag \exinj(s)}\vsWI \ownGhost{\gamma_s}{\authfull\exinj( s')}\ast\ownGhost{\gamma_s}{\authfrag \exinj(s')}$.

Next, we supply $\ownGhost{\gamma_s}{\authfrag \exinj(s')}$ and $\pure{\textit{rr}(s,x,s',y)}$ in exchange for the receipt $\Phi(y)$.
This completes the goal for the first event.

Using the induction hypothesis, we linearize remaining events in $\vv*{t}{\textlog{lin}}'$, and conclude the proof.
\end{proof}

After applying \Cref{lem:helping} and inserting the $\AU$ for the new operation, we can conclude (\ref{fig:I-atom-part-au}) and use the new state $\ownGhost{\gamma_s}{\authfull\exinj(s')}$ to complete (\ref{fig:I-atom-part-own}).

After closing the invariant, we still own $\tau\hookrightarrow^{\gamma_\pi} \function{op}(\textit{obj},x)$ and $\tau\hookrightarrow^{\gamma_\phi}\Phi$.

\subsubsection{Intermediate Steps of the Operation}\label{sec:logatom-proof-interm}
The next step is to verify $\function{op}(\textit{obj},x)$.
Using an argument in the flavor of Theorem 17 by \citet{hostert2026-completeness}, we reduce the goal to
$$\All \expr.\tau\hookrightarrow^{\gamma_\pi} \expr\wand
\tau\hookrightarrow^{\gamma_\phi}\Phi
\wand\wpre{e}{\Ret v.\tau\hookrightarrow^{\gamma_\pi} v\ast \tau\hookrightarrow^{\gamma_\phi}\Phi}$$
and use L\"ob induction to prove this goal.
To make progress, we do case analysis on whether $e\in\Val$.
The postcondition is trivially true if $e$ is a value, and for the non-value case, we apply a stronger invariant access rule~\citep[Lemma 15]{hostert2026-completeness}:

$$\mprset{fraction=--\ast}\inferhref{InvAccMaybe}{rule:inv-acc-maybe}
{\pvs[\mask_1][\mask_2]\Big(\left(\Exists K, e'.\pure{e=K[e']\land\atomic(e')}\ast\wpre{e'}[\mask_2]{\Ret v.\pvs[\mask_2][\mask_1]\wpre{K[\val]}[\mask_1]{\Phi}}\right)\\\\
\lor \pvs[\mask_2][\mask_1]\wpre{\expr}[\mask_1]{\Phi}\Big)}
{\wpre{e}[\mask_1]{\Phi}}$$
While the standard \ruleref{rule:inv-acc} rule only allows opening an invariant when the expression is atomic, this rule allows us to open an invariant for any expression and decide whether the expression is atomic \emph{afterwards}.
We have two choices after opening the invariant: either prove the first redex is atomic and close the invariant after that atomic step, or close the invariant immediately.
With the help of \ruleref{rule:inv-acc-maybe}, we can open the $I_{\textlog{atom}}$ invariant and get the latest configuration $\rho=(\vv*\pi 1\dplus(\tau,\expr_1)::\vv*\pi 2,\vv\expr,\pstate_1)$.
Next, we want to use the language-independent completeness condition~\citep[Condition 18]{hostert2026-completeness} to reduce $e$ for one step.
$$
\inferhref{LangCompleteness}{rule:lang-completeness}
{}
{\pure{\red(e_1,\pstate_1)}\ast\ownstate(\pstate_1)\proves\pvs[\mask]\Big(\left(\Exists K,e_1'.\pure{e_1=K[e_1']\land\atomic(e_1')}\ast\All \Phi.\later C_1\wand\wpre{e_1'}[\mask]{\Phi}\right)\\\\
{}\lor{} \ownstate(\pstate_1)\ast\All\Phi.\later C_2\wand\wpre{e_1}{\Phi}\Big)}$$
where
\begin{align*}
C_1\eqdef{}& \All \vv\kappa,\val_2,\pstate_2,\vv*{e}{f}.\pure{(e_1',\pstate_1)\basestep[\vv \kappa](\val_2,\pstate_2,\vv*ef)}\vsWI(\ownstate(\pstate_2)\wand\Phi(\val_2))\ast\Sep_{\expr_f\in\vv*\expr f}\wpre{e_f}{\Ret\_.\TRUE}\\
C_2\eqdef{}&\All e_2,\vv*ef.\left(\All \pstate_1'.\ownstate(\pstate_1')\vsWI[\mask]\Exists \vv\kappa,\pstate_2.\pure{(e_1,\pstate_1')\primstep[\vv \kappa](e_2,\pstate_2,\vv*ef)}\ast\ownstate(\pstate_2)\right)\vsWI[\top]\\
&\hspace*{20em}\wpre{e_2}{\Phi}\ast\Sep_{\expr_f\in\vv*\expr f}\wpre{\expr_f}{\Ret \_.\TRUE}
\end{align*}

To use this condition, we need to prove $\red(\expr_1,\pstate_1)$ and provide the ownership of the current program state $\ownstate(\pstate_1)$. The former holds because $\textlog{cfg-safe}(\rho)$, and the latter comes from the invariant.
Applying this condition leaves two cases.

For case $C_1$, we know the first redex of $e_1$ is atomic, so we can keep the invariant open for one step.
In fact, $\All \Phi.\later C_1\wand\wpre{e_1'}[\mask]{\Phi}$ from the condition is the exact rule needed to justify that step.
After that step, we are obligated to prove%
\footnote{In the goal, assume $\top\setminus\namesp$ is the mask after opening the $I_{\textlog{atom}}$ invariant.}
$$\left(\ownstate(\pstate_2)\wand\pvs[\top\setminus\namesp][\top]\wpre{K[\val_2]}{\Phi}\right)\ast\Sep_{\expr_f\in\vv*\expr f}\wpre{e_f}{\Ret\_.\TRUE}$$
under the assumption that $(e_1',\pstate_1)\basestep[\vv \kappa](\val_2,\pstate_2,\vv*ef)$.

Using \ruleref{rule:wp-resolve-trace-interm}, we update $\prophtrace{(\vv*\pi 1\dplus(\tau,\expr_1)::\vv*\pi 2,\vv\expr,\pstate_{1})}{p}{\pi_\Sigma(\anntr)}$ to $\prophtrace{(\vv*\pi 1\dplus(\tau,K[\val_2])::\vv*\pi 2,\vv\expr\dplus\vv*\expr f,\pstate_{2})}{p}{\pi_\Sigma(\anntr')}$, where $\anntr'\prefixof\anntr$.
We then update the ghost state $\gamma_\pi$ and $\gamma_e$ accordingly; this includes allocating new ghost points-tos for each thread in $\vv*\expr f$: $$\Sep_{i=0}^{\left|\vv*\expr f\right|-1}(\left|\vv\expr\right|+i)\hookrightarrow^{\gamma_e} \vv*\expr f(i)$$

We now split the goal into two parts: $\ownstate(\pstate_2)\wand\pvs[\top\setminus\namesp][\top]\wprename\cdots$ and the iterated separating conjunction of $\wprename$\uq s over $\vv*ef$.

For part 1, we obtain the ownership of the new state after that atomic step, and we need to close the invariant.
Most parts are mechanical, and the only interesting part is (\ref{fig:I-atom-part-au}), for which we need to prove $\textit{interpret}(\anntr,\tau,\Phi)\proves\textit{interpret}(\anntr',\tau,\Phi)$.
If $\pi_\tau(\anntr')$ is not nil, because $\anntr'\prefixof\anntr$, $\pi_\tau(\anntr)$ must also be non-nil and have the same head as $\pi_\tau(\anntr')$, which implies $\textit{interpret}(\anntr',\tau,\Phi)=\textit{interpret}(\anntr,\tau,\Phi)$.
Otherwise, if $\pi_\tau(\anntr')$ is nil, then $\textit{interpret}(\anntr',\tau,\Phi)$ is trivially true.

After closing the invariant, we can use the L\"ob induction hypothesis to finish part 1.

Part 2 is proved by \Cref{lem:safe-child} below.

For case $C_2$, we do not know whether the first redex of $e_1$ is atomic, so we have to immediately close the invariant.
As a result, we lose all the knowledge learned from the invariant, including the current configuration $\rho$ and the ownership of the current state $\ownstate(\pstate_1)$.
Fortunately, with the help of $\All \Phi.\later C_2\wand\wpre{e_1}{\Phi}$, we can still take one step with the invariant closed.
To prove $C_2$, we need to prove
$$\wpre{e_2}{\Phi}\ast\Sep_{\expr_f\in\vv*\expr f}\wpre{e_f}{\Ret\_.\TRUE}$$
under the assumption that
$$\All \pstate_1'.\ownstate(\pstate_1')\vsWI[\mask]\Exists \vv\kappa,\pstate_2.\pure{(e_1,\pstate_1')\primstep[\vv \kappa](e_2,\pstate_2,\vv*ef)}\ast\ownstate(\pstate_2)$$
This is a rather strong assumption because it allows us to take one step from $e_1$ under \emph{any} state.

To use this assumption, we open the invariant again and get a (possibly different) configuration  $\rho'=(\vv*\pi 1\dplus(\tau,\expr_1)::\vv*\pi 2,\vv\expr,\pstate_1')$ and the ownership of $\ownstate(\pstate_1')$.
The remaining proof for this case is identical to case $C_1$, which also uses \Cref{lem:safe-child} to discharge the goal for threads in $\vv*ef$.

\begin{lemma}\label{lem:safe-child}
$\knowInv{}{I_{\textlog{atom}_\mu}(\gamma,\textit{obj},p)}
\ast\ownGhost{\gamma}{\Square(\gamma_\pi,\gamma_e,\gamma_\phi,\gamma_s)}\proves i\hookrightarrow^{\gamma_e}\expr\wand \wpre{e}{\Ret \_.\TRUE}$
\end{lemma}
\begin{proof}[Proof Sketch]
Similar to \Cref{sec:logatom-proof-interm}, the proof is based on L\"ob induction, but this time, we use \ruleref{rule:wp-resolve-trace-child} to update the \textlog{prophTrace} resource.
\end{proof}

\subsubsection{The Return of the Operation}
As the final step, we need to verify $\ghostcode{\Resolve \prophid to \Inr(\tau,y)}$ under the precondition that
$\tau\hookrightarrow^{\gamma_\pi}y$ and $\tau\hookrightarrow^{\gamma_\phi}\Phi$.
The proof is symmetric to \Cref{sec:logatom-proof-call}: we open the invariant, use \ruleref{rule:wp-resolve-trace-ret} to update the prophecy variable, split the new trace into $\vv*t{\textlog{lin}}$ and $\anntr'$, update ghost states, use \Cref{lem:helping} to linearize operations in $\vv*t{\textlog{lin}}$, and close the invariant.
But this time, for (\ref{fig:I-atom-part-au}), we extract the receipt $\Phi'(y)$ from the iterated separating conjunction, where $\Phi'=\Phi s(\tau)$. Since we still own the fragment $\tau\hookrightarrow^{\gamma_\phi}\Phi$ recording this operation's receipt type, its agreement with the authoritative receipt map $\ownGhost{\gamma_\phi}{\authfull\Phi s}$ yields $\later(\Phi\equiv\Phi')$.

We are now left with proving $\later(\Phi\equiv\Phi')\ast\Phi'(y)\proves\wpre{();y}{\Phi}$.
Take a pure step and apply \ruleref{rule:wp-value}, and we get $(\Phi\equiv\Phi')\ast\Phi'(y)\proves{\Phi}(y)$, which is trivially true.

This concludes the proof of the specification for $\textit{wrapo}(\function{op})$.

\section{Transporting Logical Atomicity through Logical Refinement}\label{sec:refine}

Contextual refinement is another commonly used approach to specifying concurrent data structures.
For example, with contextual refinement, one can capture that a fine-grained concurrent data structure behaves in an atomic way by proving that it refines a coarse-grained implementation where the operations are guarded by a global lock.
Such a contextual refinement implies that any behaviors that a client of the fine-grained implementation can observe are a subset of the behaviors that they could observe with the coarse-grained implementation.
Because the latter behaves atomically by construction, this means the former also appears to do so from the client's perspective.

Since the coarse-grained implementation is atomic by construction, one can naturally also prove a logically atomic triple specification for it.
One might expect that the contextual refinement would also imply that this logically atomic specification must hold for the fine-grained implementation as well.
However, there has previously been no way to go from a refinement to a logically atomic triple.
This is a notable gap, since in contrast, prior work \emph{has} established equivalences between contextual refinement and linearizability~\citep{filipovic2010-ctx-refine}.

Using our results from previous sections, we close this gap by showing how to transport a logically atomic triple across a refinement specification.
Specifically, we target ReLoC \cite{frumin2021-reloc}, a relational logic defined on top of Iris \cite{IrisGroundUp} for proving contextual refinement of \heaplang{} programs.
ReLoC introduces a \emph{logical refinement judgment} $\refinesS{e}{e'}{A}$ where $A$ is an Iris relation that holds on the return values of $e$ and $e'$.
For example, if $A \eqdef{} (=)$, \ie the Leibniz equality relation, then $e$ and $e'$ always return equal values.
ReLoC's soundness theorem shows that logical refinement implies contextual refinement.

We show that ReLoC's notion of logical refinement transports logically atomic specifications.
\begin{theorem}[Refinement Transports Logical Atomicity]\label{thm:atomic-refines}
  If
\begin{enumerate}
\item  $\nofork{\function{init}()}{\emptyset}$ and $\nofork{\function{init\uq}()}{\emptyset}$,
\item $\atomspec[\mu]{\IsP'}{\PState'}{\function{init\uq}}{\function{op\uq}}$,
\item $\refinesS{\function{init}()}{\function{init\uq}()}{A}$, and
\item $A(\textit{obj}, \textit{obj}') \ast \pure{\textit{vr}(x)} \proves \refinesS{\function{op}(\textit{obj},x)}{\function{op\uq}(\textit{obj}',x)}{(=)}$ for all $\textit{obj}$, $\textit{obj}'$, and $x$,
\end{enumerate}
then $\atomspec[\mu]{\IsP_{\mu}}{\PState_{\mu}}{\textit{wrapi}(\function{init})}{\textit{wrapo}(\function{op})}$.
\end{theorem}
To use this theorem, a user picks a relation $A$ that relates the result of $\function{init}$ and $\function{init\uq}$.
This is typically an invariant that relates the internal states of the data structures on the two sides of the refinement.
Then given $A$-related data structures $\textit{obj}, \textit{obj}'$ and a valid request $x$, the user has to show that ${\function{op}(\textit{obj},x)}$ refines ${\function{op\uq}(\textit{obj}',x)}$ and they return equal results.
As a result, if $(\function{init\uq},\function{op\uq})$ satisfies a logically atomic specification, then so does $(\function{init},\function{op})$.

To establish our transport theorem, we first show that logical refinement can be used to transport linearizability.
\Cref{thm:atomic-refines} then follows immediately from \Cref{thm:logatom-soundness} and \Cref{thm:logatom}.

\begin{lemma}[Refinement Transports Linearizability]\label{thm:lin-refines}
  If
  \begin{enumerate}
  \item  $\nofork{\function{init}()}{\emptyset}$ and $\nofork{\function{init\uq}()}{\emptyset}$,
  \item $\lin[\mu]{\function{init\uq}}{\function{op\uq}}$,
  \item $\refinesS{\function{init}()}{\function{init\uq}()}{A}$, and
\item $A(\textit{obj}, \textit{obj}') \ast \pure{\textit{vr}(x)} \vdash \refinesS{\function{op}(\textit{obj},x)}{\function{op\uq}(\textit{obj}',x)}{(=)}$ for all $\textit{obj}$, $\textit{obj}'$, and $x$,
  \end{enumerate}
  then $\lin[\mu]{\function{init}}{\function{op}}$.
\end{lemma}
The fact that this lemma should be true is not surprising, because, as alluded to above, prior work by \citet{filipovic2010-ctx-refine} has established correspondences between contextual refinements and linearizability, and we know that logical refinement implies contextual refinement.
However, that prior work dealt with a considerably simpler first-order programming language that lacked dynamic allocation.
Thus, we cannot directly appeal to prior results.

The key idea behind our proof of this lemma is to construct a program called the \emph{most general client} (MGC) that covers all possible behaviors permitted by the most adversarial steps relation.
We then use refinement to transport the MGC's interaction with ${\function{op}}$ into an interaction with ${\function{op\uq}}$.
Similar most general client constructions are used elsewhere in the linearizability and refinement literature.
The idea is that this client repeatedly invokes concurrent operations with non-deterministic arguments.
This non-determinism ensures that for any sequence of invocations in the most adversarial step relation, the most general client could trigger a similar sequence of concurrent calls.

One subtlety is that, unlike the way this program is constructed in some prior work, we \emph{cannot} construct a single most general client program in \heaplang{}: the language does not have primitives that can non-deterministically generate all possible valid arguments to invoke the operation with.
We instead define a function $\textit{MGC} : \Sigma^* \rightarrow \Val$ at the meta-level, such that given a trace $\phytr$, the program $\textit{MGC}(\phytr)$ invokes concurrent operations in the pattern specified by $\phytr$.
In particular, this $\textit{MGC}(\phytr)$ is a \heaplang{} function that takes an operation $\function{op}$ and an object $\textit{obj}$ as arguments.

Concretely, to show that $(\function{init},\function{op})$ is linearizable, we are given some execution of $\function{init}$ returning an object $\textit{obj}$, along with $\rho \mathrel{\smash{\advsteps[\textit{obj}.\function{op}/\textit{vr}]{\vv t}}} \rho'$, and we have to show that $\vv t$ is a linearizable trace.
By construction, there is an execution of $\textit{MGC}(\phytr)(\function{op}, \textit{obj})$ that exhibits the same behavior and order of events as this adversarial interaction.
From the assumption that $\function{init}$ refines $\function{init\uq}$, we know that there is a corresponding execution of $\function{init\uq}$ returning some object $\textit{obj\uq}$ such that $\textit{obj}$ refines $\textit{obj\uq}$.
We then show that ${\textit{MGC}(\phytr)(\function{op}, \textit{obj})}$ refines ${\textit{MGC}(\phytr)(\function{op\uq}, \textit{obj\uq})}$, which in turn gives us an execution trace for $\textit{MGC}(\phytr)(\function{op\uq}, \textit{obj\uq})$.
We conclude by converting this MGC execution into a most adversarial interaction for $(\function{init\uq}, \function{op\uq})$ with trace $\vv t$ and applying the assumption that $(\function{init\uq},\function{op\uq})$ is linearizable.

In order to make the argument in the previous paragraph work, we must ensure (1) that the MGC indeed fully generates equivalent traces, and (2) that a refinement between MGCs implies that the traces are equivalent.
The latter involves another slight technical complication: contextual refinement is an extensional property and is only concerned with the externally observable behavior; however, in \heaplang{}, calls and returns are \emph{not} externally observable, because
functions are first-class values and calls are just beta-reductions of $\lambda$ terms.

To deal with this, we have the MGC keep track of call/return events in a physical log stored at some reference cell $\loc_{\textit{tr}}$.
The ${\textit{MGC}(\phytr)(\function{op}, \textit{obj})}$ construction forks a new thread for each $\textlog{Call}(x)$ operation in the trace $\vv t$, with shared access to this reference cell.
Each thread invokes $\function{op}$ on $\textit{obj}$ with request $x$, but it also ``emits'' an event (containing the request) by adding the event to the physical log stored in the reference cell.
Similarly, after the function returns, the result is also added to the log.
By dereferencing $\loc_{\textit{tr}}$ at the end of $\textit{MGC}(\phytr)$, we make the trace events physically observable, and as a consequence, the MGC refinement implies trace refinement.
Additional details on the MGC construction and supporting lemmas appear in \Cref{sec:app:mgc}.

\section{Embedding External Linearizability Techniques}\label{sec:external}

Besides formally characterizing the expressivity of the logic, our completeness result (\Cref{thm:logatom}) also provides a new approach for obtaining a logically atomic specification in Iris: we can verify linearizability of a data structure using any technique external to Iris and then formally embed the result as a logically atomic specification.

Among the vast array of linearizability proof methods in the literature, we reformulate three representative methods, aspect-oriented proofs for queues~\citep{henzinger2013-aspect}, forward simulation with commit points~\citep{bouajjani2017-forward}, and meta-configuration tracking~\citep{jayanti2024-metaconfig}, using the approach below.
As a first step, we mechanize in Rocq the meta-theorems from these prior works showing that these methods indeed imply linearizability.
Since these techniques have already been described in the literature, after mechanizing the statement of these meta-theorems, we were able to discharge most of the proofs using Large Language Models (LLMs).

Each of these methods requires showing some conditions on the trace of events that a data structure can generate.
So, to apply these methods, one needs a way to prove that the traces generated by an implementation of a data structure satisfy these conditions.
In our case, we will use Iris itself to prove that these trace conditions hold.
This means we can still take advantage of Iris's expressive features for modular reasoning about concurrent state manipulation, but we avoid having to explicitly reason about linearization points and atomic updates.
Instead, we just prove the trace property holds, apply the meta-theorems to deduce linearizability, and thereby obtain logical atomicity.

To leverage Iris for this task, we adapt a \emph{free theorem} result that \citet{birkedal2021-free} established for Iris.
Their result allows one to deduce trace properties from the separation logic specifications of a library.
Their approach models a trace as a separation logic resource, and puts this resource in an appropriate Iris invariant, so that the usual soundness theorem of Iris implies properties about the resulting trace that a program can generate.

In the remainder of this section, we start by adapting their theorem to account for our adversarial-steps trace reduction relation, and then we describe each of the external linearizability proof methods we have formulated in this manner.
The following section~(\Cref{sec:case-studies}) describes how we then use these methods to get logically atomic triples for various example data structures.

\subsection{Deriving Trace Invariants from Separation Logic Specifications}\label{sec:free-theorem}
\citet{birkedal2021-free} propose a theorem to get a trace property about a data structure from its separation logic specification.
Their key idea is adding a wrapper around data structure operations to emit a call-return trace via two special primitive operations \langkw{emit} and \langkw{fresh}.
Because they design the specifications of these operations in a way that ensures the emitted trace must satisfy a predicate $I\colon\Sigma^*\to\mProp$, every trace produced by an interaction between the data structure and a well-formed client that uses this specification must also belong to $I$.

In our case, we want to talk about the traces that would be generated under the most adversarial steps relation, in order to connect to linearizability.
Thus, we reuse the most general client (MGC) construction from the proof of \Cref{thm:lin-refines}: a client that spawns one thread for each operation in a trace $\phytr$.
Instead of using primitive $\langkw{emit}$ and $\langkw{fresh}$ commands, we simulate a trace by storing these events in a physical reference cell.

As we discussed in \Cref{sec:refine}, executions of the MGC cover all behaviors captured by the $\smash{\advsteps{}}$ relation.
Consequently, to derive a property of every adversarial trace, it suffices to prove a triple about every instance of the MGC using a data structure's operations.

That in turn only requires proving a specification we call $\tracespec[I]{\MgcInv}{\textit{vr}}{\function{init}}{\function{op}}$ about $\function{init}$ and $\function{op}$ wrapped with the trace manipulating code.
This specification asserts that, for some invariant $\MgcInv$ chosen by the prover, (1) initializing the data structure establishes $\MgcInv$, and (2) for a valid request, the wrapped operation is safe given the trace invariant. Proving this requires showing, at each call to $\function{emit}$ and $\function{fresh}$, that the appended event preserves $I$.
Chaining everything together, we have this theorem.
\begin{theorem}\label{thm:bdgjst-free}
  If\/ $\tracespec[I]{\MgcInv}{\textit{vr}}{\function{init}}{\function{op}}$, then
  \begin{enumerate}
  \item $\safe{\function{init}()}{\varnothing}$ and
  \item if $([\function{init}()],\varnothing)\tpsteps([\textit{obj}],\pstate)$ and $(\nil,\nil,\pstate) \mathrel{\smash{\advsteps[\textit{obj}.\function{op}/\textit{vr}]{\phytr}}}(\vv{\pi}, \vv{e}, \pstate')$, then every expression in $\vv{\pi}$ and $\vv{\expr}$ is not stuck, and $I(\phytr)$.
  \end{enumerate}
\end{theorem}
Additional details about the lemmas and proof of this theorem are included in \Cref{sec:app:mgc}.
Each linearizability proof method we consider in the remainder of this section requires showing conditions on a data structure's traces that can be formulated in terms of an appropriate predicate $I$ in the above theorem.
Thus, we can use Iris to prove that these conditions hold.

\subsection{Aspect-Oriented Proofs}
The idea behind aspect-oriented proofs of linearizability is to find, for a given class of data structures such as queues or stacks, a collection of properties about traces that imply linearizability for any implementation of those data structures.
For example, the original result by~\citet{henzinger2013-aspect} establishes such a theorem for queues saying that if for every completion of an adversarial call-return trace, there exists an injective partial map $M$ from dequeue (pop) operations to enqueue (push) operations, and $M$ satisfies some consistency properties, then the data structure is linearizable as a queue.
Some examples of these consistency properties are
\begin{itemize}
\item \emph{Valid Match.} If a dequeue $\tau_2$ is mapped to an enqueue $\tau_1$, then the response of $\tau_2$ must equal the request of $\tau_1$, and $\tau_2$ must not happen before $\tau_1$ (\ie $\tau_2$ must not finish before $\tau_1$ starts).
\item \emph{Ordered Correct.} For two pairs of enqueue and dequeue operations $(M(\tau_1),\tau_1)$ and $(M(\tau_2),\tau_2)$, if $M(\tau_1)$ happens before $M(\tau_2)$, then $\tau_2$ must not happen before $\tau_1$.
\item \emph{Empty Sound.} If a dequeue reports that the queue is empty, then the queue must be logically empty at some point between the call and return of the dequeue.
\end{itemize}

Later, \citet{dodds2015-ts-stack} formulated an analogous theorem for stacks, and \citet{DBLP:journals/pacmpl/OhmanN22} developed similar results for snapshottable arrays.

We have adapted \citeauthor{henzinger2013-aspect}'s queue theorem to our setting. %
To state our result, we first need to define a sequential specification $\mu_{\textit{queue}}$ for queues, following the style of $\mu_{\textit{bag}}$ in \Cref{example:bag}:
\begin{align*}
\mu_{\textit{queue}}.s_0\eqdef{}&\nil\hspace{4em}
 \mu_{\textit{queue}}.\textit{vr}\eqdef\mu_{\textit{bag}}.\textit{vr}\\
 \mu_{\textit{queue}}.\textit{rr}(s,x,s',y)\eqdef{}&(\Exists v.x=\Some(v)\land s'=s\dplus[v]\land y=\TT)\\
 {}\lor{}&x=\None{}\land s=\nil\land s'=\nil\land y=\None\\
 {}\lor{}&x=\None{}\land\Exists v. s=v::s'\land y = \Some(v)
\end{align*}

Then, our formulation of their theorem says
\begin{theorem}
For a physical trace $\phytr$, if\/ $\textlog{well-formed-bag}(\phytr)$ and $I_{\textit{queue}}(\phytr)$, then $\trlin[\mu_{\textit{queue}}]{\phytr}$, where
\begin{align*}
\textlog{well-formed-bag}(\phytr)\eqdef{}&\All\tau.(\Exists v.\pi_\tau(\phytr)\prefixof[(\tau,\textlog{Call}(\Some v)); (\tau,\textlog{Ret}(\TT))])\\
&\mathmakebox[\widthof{\ensuremath{\All\tau.}}][r]{{}\lor{}}(\Exists v.\pi_\tau(\phytr)\prefixof[(\tau,\textlog{Call}(\None)); (\tau,\textlog{Ret}(\Some v))])\\
&\mathmakebox[\widthof{\ensuremath{\All\tau.}}][r]{{}\lor{}}\pi_\tau(\phytr)\prefixof[(\tau,\textlog{Call}(\None)); (\tau,\textlog{Ret}(\None))]\\
I_{\textit{queue}}(\phytr)\eqdef{}&\Exists M.\textlog{valid-match}(\phytr,M)\land \textlog{empty-sound}(\phytr,M)\\
&\mathmakebox[\widthof{\ensuremath{\Exists M.}}][r]{{}\land{}}\textlog{order-matched}(\phytr,M)\land\textlog{order-unmatched}(\phytr,M)%
\end{align*}
\end{theorem}
$I_{\textit{queue}}$ mostly follows the consistency properties required by \citeauthor{henzinger2013-aspect}, but is tweaked to account for partial correctness.
Notably, unlike the original version that relies on a completion of the trace, our $I_{\textit{queue}}$ can be incrementally verified: the empty list satisfies $I_{\textit{queue}}$, and knowing that $I_{\textit{queue}}(\phytr)$, verifying $I_{\textit{queue}}(\phytr\dplus u)$ only requires checking the new event $u$ is compatible with $\phytr$.
This requires slightly strengthening \textlog{valid-match}.
In addition to the requirement by \citeauthor{henzinger2013-aspect}, our \textlog{valid-match} further requires that for enqueues $\tau_1$ and $\tau_2$, if $\tau_1$ happens before $\tau_2$ and $\tau_2$ is matched with some dequeue by $M$, $\tau_1$ must also be matched by $M$.
This additional restriction forces the user to construct $M$ incrementally along with the execution of the queue, instead of delaying matching $\tau_1$ with a dequeue forever.

\subsection{Forward Simulation with Commit Points}
In \Cref{sec:refine} we discussed methods of transporting linearizability (and logical atomicity) by establishing a refinement between an implementation $M$ of some data structure and another implementation $M'$, and then showing that $M'$ is linearizable.
However, one can also use refinement to transport linearizability even when $M'$ is not another code-level implementation but is instead an abstract state machine or labeled transition system (LTS).
In particular, if one can prove a trace refinement between $M$ and a transition system $M'$, and the traces of the transition system are all linearizable, then $M$ is linearizable.
In many cases, we can take $M'$ to simply be an LTS which performs the operations in atomic transitions according to the sequential specification.
This $M'$ is then linearizable by definition.

Then, to establish the refinement, one approach is to construct a forward simulation between $M$ and $M'$.
Such forward simulation proofs are natural to carry out inside of a concurrent separation logic such as Iris, because they can be constructed as we reason over the steps of $M$.
However, in some cases, it is not possible to directly construct a forward simulation between $M$ and an atomic transition system $M'$.
In particular, if $M$ has future-dependent linearization points, then such a forward simulation is not possible.

\citet{bouajjani2017-forward} propose an alternate technique to work around this.
Specifically, for proving linearizability of queues, they define an abstract state
machine $\textit{AbsQ}$ that only maintains a partial order of enqueues rather than a full sequential order.
This makes it possible to prove a forward simulation between a queue implementation and $\textit{AbsQ}$, even for queues that have future-dependent enqueues.
They then show that $\textit{AbsQ}$ is linearizable \wrt the standard queue model using a \emph{backward} simulation proof.
Thus, establishing a forward simulation to $\textit{AbsQ}$ implies linearizability.
\citeauthor{bouajjani2017-forward} also develop a similar abstract state machine for stacks.

Following their approach, we mechanize a definition of $\textit{AbsQ}$ as an LTS and formalize their key theorem:
\begin{theorem}
For a physical trace $\phytr$, if there exists an annotated trace $\anntr$ that forward simulates $\textit{AbsQ}$ and $\phytr = \pi_\Sigma(\anntr)$, then $\trlin[\mu_{\textit{queue}}]{\phytr}$.
\end{theorem}
The existence of such an annotated trace that forward simulates $\textit{AbsQ}$ can be established using an appropriate trace invariant in Iris.

\subsection{Meta-Configuration Tracking}
\citet{jayanti2024-metaconfig} note that when doing forward reasoning, if instead of maintaining one abstract state of the data structure, one considers a set of possible abstract states, one obtains a complete forward-reasoning-based linearizability proof technique.
They call the set of states a meta-configuration and hence name their technique meta-configuration tracking.
Unlike the two techniques above, this technique is universal in the sense that it is not specific to a particular class of data structures, such as queues or stacks.

In our mechanization, we define a single configuration as a pair of the form $(s, \textit{OpSt})$ with type $S\times(\textdom{Tag}\fpfn(\Val\times\Val^?))$, where $s$ represents the state of the object, and $\textit{OpSt}$ is a map from operation invocation tags to a \emph{status}.
There are two possible statuses for an operation:
\begin{enumerate}
\item called with $\val$ but not linearized, which is represented as $(\val,\bot)$; or
\item called with $\val$ and linearized with $\valB$, which is represented as $(\val,\valB)$.
\end{enumerate}
A meta-configuration is defined as a set of single configurations.
We define a transition relation on meta-configurations $M_1\mconfigsteps[\mu]{\phytr}M_2$ with three types of atomic transitions:
\begin{enumerate}
\item an operation $\tau$ is called with value $\val$, which emits event $(\tau,\textlog{Call}(\val))$ and adds a new operation $\mapsingleton{\tau}{(\val,\bot)}$ for \emph{every} configuration in the set;
\item an operation $\tau$ in a configuration $(s,\mapinsert{\tau}{(\val,\bot)}{\textit{OpSt}})$ possibly linearizes with $\valB$, which does not emit any event but inserts a new configuration $(s',\mapinsert{\tau}{(\val,\valB)}{\textit{OpSt}})$ satisfying $\textit{rr}(s,\val,s',\valB)$; and
\item an operation $\tau$ returns with $\valB$, which emits event $(\tau,\textlog{Ret}(\valB))$ and prunes every single configuration whose status of $\tau$ is not $(\_,\valB)$.
\end{enumerate}
Then, we have the following theorem:
\begin{theorem}
For a physical trace $\phytr$, if $\textlog{well-formed}(\phytr)$ (as defined in \cref{eq:well-formed}) and there exists a non-empty meta-configuration $M$ such that $\set{(s_0,\varnothing)}\mconfigsteps[\mu]{\phytr} M$, then $\trlin[\mu]{\phytr}$.
\end{theorem}
Again, such a meta-configuration can be constructed in an Iris Hoare triple proof using ghost state, allowing us to apply this method within Iris.

\section{End-to-End Logical Atomicity Results}
\label{sec:case-studies}
As end-to-end case studies, we prove the linearizability of several queue algorithms using techniques presented in \Cref{sec:refine} and \Cref{sec:external}, and obtain logically atomic triples for them using \Cref{thm:logatom}.
Since proofs of linearizability using these techniques have already been described in the literature, either with pencil and paper or mechanically using a different framework, writing these proofs using our reformulation requires little human effort, and substantial parts of the Rocq proofs for several of these examples were automated using LLMs.

\paragraph{Herlihy-Wing Queue}
The Herlihy-Wing queue~\cite{herlihy1990-linearizability} is a famous lock-free queue algorithm known for its future dependency, and is a standard benchmark for linearizability proof techniques.
\citet{jung2019-proph} mechanized a proof of logical atomicity for the Herlihy-Wing queue using prophecy variables in Iris, and an implementation of this data structure in \heaplang{} is maintained as part of the Iris project.
We prove the linearizability of this implementation in three different ways, using the aspect-oriented, forward simulation, and meta-configuration tracking methods described in the previous section, and obtain logically atomic triples that have the same form as the ones proven by \citet{jung2019-proph}.
For the proof of linearizability using each technique, we verify a trace invariant using \Cref{thm:bdgjst-free}, and then conclude this trace is linearizable using the technique-specific theorem.
None of these linearizability proofs use prophecy variables.
Notably, the aspect-oriented proof is 27\% shorter than the direct Iris proof (2004 LoC vs 2763 LoC), though this does not include the proof of the aspect-oriented queue meta-theorem.

\paragraph{Baskets Queue}
The Baskets Queue is another lock-free queue~\cite{DBLP:conf/opodis/HoffmanSS07}.
Like the Herlihy-Wing queue, it features a challenging form of future-dependent linearization in enqueue operations.
\citet{bouajjani2017-forward} mention it as one example for which their forward simulation proof technique should be applicable.
We implement the Baskets Queue in \heaplang{}, prove its linearizability using their forward simulation, and obtain a logically atomic specification for it.
Since \heaplang{} is designed as a garbage-collected language, our implementation of the algorithm removes several aspects of the original implementation designed for dealing with manual reclamation and avoiding the ABA-problem~\citep{DBLP:journals/jpdc/MichaelS98}.
Nevertheless, it still retains the challenging future dependency in enqueues.
To our knowledge, this is the first machine-checked proof of linearizability for the Baskets Queue.

\paragraph{Folly MPMC Queue}
The multi-producer-multi-consumer (MPMC) queue from Meta's C++ library Folly \cite{folly} is a high-performance bounded concurrent queue.
The queue is optimized, production-ready, and it scales to thousands of consumer and producer threads.

\citet{DBLP:conf/cpp/VindumFB22} implement the MPMC queue in \heaplang{} and show using ReLoC that it contextually refines a coarse-grained queue.
A key challenge of the refinement proof is the queue's use of external linearization points.
Using \Cref{thm:atomic-refines}, we directly reuse their refinement proof to derive a logically atomic specification for the MPMC queue.
Proving the logically atomic specification for the coarse-grained queue is straightforward.

\section{Related Work}

\paragraph{Linearizability Proof Techniques}
As discussed in \Cref{sec:external}, one of the benefits of our completeness result is that it allows us to use external linearizability proof techniques in order to establish logical atomicity inside of Iris.
This is beneficial because, as we alluded to, there is a vast array of such techniques that have been developed in the research literature.
Much of the focus of the research community has been specifically developing methods that work for data structures that exhibit challenging features, such as helping and future-dependent linearization points.
We have already used and discussed several of these, including aspect-oriented linearizability~\citep{henzinger2013-aspect}, a forward simulation refinement technique~\citep{bouajjani2017-forward}, and meta-configuration tracking~\citep{jayanti2024-metaconfig}.
Other well-known methods that have found repeated use include simulation proofs with canonical automata~\citep{DBLP:journals/entcs/ColvinDG05} and hindsight reasoning~\citep{DBLP:conf/podc/OHearnRVYY10}.
\citet{dongol2015-linear-survey} provide a survey of many other techniques, along with what examples they have been applied to and known limitations.
It would be interesting to similarly mechanize these other proof methods and use them in combination with our results.

\paragraph{Program Logics for Linearizability}
A number of program logics are designed to directly establish linearizability, rather than trying to internalize atomicity in the logic.
The RGSep logic of \citet{DBLP:conf/concur/VafeiadisP07} was used in this way to prove linearizability by instrumenting code with linearization-point actions and establishing a simulation~\citep{DBLP:conf/vmcai/Vafeiadis09, vafeiadis2010-cave}.
The RGSim~\citep{DBLP:conf/popl/LiangFF12, DBLP:conf/pldi/LiangF13} framework handles helping and future-dependent linearization points.
Its soundness theorem yields a contextual refinement that is equivalent to linearizability.
\citet{khyzha2017-partial-order} present a logic that allows for incrementally constructing a partial order over operations rather than a single linearization, allowing ordering decisions to be delayed for data structures with future-dependent linearizations.
Several program logics have been constructed around the hindsight reasoning method mentioned above, which involves some form of temporal reasoning~\citep{DBLP:conf/podc/OHearnRVYY10, meyer2022-plankton, DBLP:journals/pacmpl/0001W023}.

As mentioned earlier, one limitation of working with linearizability inside of a program logic is that it is often not possible to compose such linearizability specifications inside of the program logic, a limitation that motivated the development of logically atomic specifications.
However, a line of work has sought to reconstruct and generalize linearizability in a way that is more directly compositional.
\citet{DBLP:journals/pacmpl/ValeSC23} use concurrent game semantics to derive such a generalized form of linearizability and give a program logic that is sound with respect to this generalized notion.

Most recently, \citet{hatti2026-lhl} have developed a rely-guarantee logic that is both sound and complete for a version of this compositional form of linearizability.
They achieve completeness by incorporating a form of \emph{possibility reasoning}, which maintains a set of possible linearizations, in contrast to our use of prophecy variables.
Their result also generalizes several variants of linearizability, including atomic, set, and interval linearizability.
In comparing their approach to logical atomicity, they note that
``[i]t is known that a logically atomic triple implies Herlihy-Wing linearizability\dots but there is no evidence that logical atomicity can be used to specify all Herlihy-Wing
linearizable objects, so its completeness is unknown.''
Our results close this gap by showing that in fact logical atomicity is complete for linearizability.

\paragraph{Logical Atomicity}
\citet{DBLP:conf/popl/JacobsP11} introduced an approach to internalizing atomicity in a program logic by parameterizing specifications with a client-supplied action that updated ghost state.
The operation would invoke this action at the point where the operation logically took effect.
HOCAP~\citep{svendsen2013-hocap} developed this higher-order style further.
TaDA~\citep{da-rocha-pinto2014-tada} introduced a distinct \emph{atomic triple} judgment, though in its original form it could not reason about helping.
Iris~\citep{Iris1} showed how to encode the logically atomic triples of TaDA inside of a higher-order framework using atomic update predicates
 that are similar to the ghost updates in \citeauthor{DBLP:conf/popl/JacobsP11}'s framework.

None of these frameworks for logically atomic specifications could initially deal with future-dependent linearization points.
This was resolved by introducing prophecy variables to Iris~\citep{jung2019-proph}, which are essential for our completeness proof.
\citet{jung2019-proph} also noted that for many forms of future-dependent linearization, a stronger form of $\ghostcode{\langkw{resolve}}$ called \emph{atomic resolve} is needed to prove logical atomicity, but our work suggests that the atomic resolve is actually unnecessary for completeness.
One of the motivations for later credits~\citep{spies2022-later} was for reasoning about a form of helping in logical atomicity proofs that \citet{spies2022-later} called \emph{unsolicited helping}.
But again, for completeness, later credits turn out to be unnecessary.

Using completeness, our work has shown how to embed other linearizability proof methods in Iris in order to obtain logically atomic specifications.
Prior to our work, \citet{DBLP:conf/ecoop/PatelSW24} encoded a form of the hindsight reasoning linearizability proof method into Iris.
Their encoding instruments a program with prophecy variables and uses these, along with helping, to derive logically atomic specifications from a user-supplied hindsight proof.
Our proof uses a similar combination of prophecies with helping to derive logical atomicity from any linearizability proof.

Although we have focused on the correspondence between logical atomicity and linearizability, there are data structures that can be specified using logically atomic triples that fall outside the scope of linearizability.
For example, \citet{da-rocha-pinto2014-tada} observe that atomic triples can capture behaviors that the sequential specifications used in linearizability cannot, because the atomic triple can constrain how a client may invoke operations through expressive preconditions.

An alternative to logical atomicity has been pursued by the FCSL logic~\citep{nanevski2014-auth-monoid}, based on using Hoare triples over time-stamped histories tracked as ghost state.
\citet{sergey2016-nonlin} advocate such history-tracking specifications as a compositional correctness condition in their own right, showing how to specify certain classes of data structures and properties that lie outside the scope of linearizability specifications.

\paragraph{Automated Linearizability Provers}
Many methods for establishing linearizability have been incorporated into automated tools.
\citet{vafeiadis2010-cave} developed CAVE, a tool that automatically proves linearizability by searching for linearization points.
CAVE cannot handle general future dependency.
However, the original work on aspect-oriented linearizability showed that CAVE could be used to check some of the conditions needed to apply the queue aspect-oriented theorem.
\citet{zhu2015-poling}'s Poling tool goes beyond CAVE and is able to handle helping.
More recently, \citet{meyer2022-plankton} and \citet{meyer2023-nekton} develop separation logic-based tools that use hindsight reasoning.
Although our examples have used proof techniques that can be done with mechanized interactive proofs in Rocq, it would be interesting to try to connect to automated techniques for establishing linearizability as a method for incorporating logically atomic specifications in Iris.

In particular, it would be interesting to compare the automation achievable using tools designed for linearizability, as opposed to automated or semi-automated techniques targeting logical atomicity directly.
Several works have taken this latter route.
Voila~\citep{wolf2021-voila} is a proof outline checker for the TaDA logic that checks a program annotated with key steps.
\citet{mulder2023-diaframe} extend the Diaframe proof automation library for Iris to handle logical atomicity and contextual refinement. They achieve full automation on simpler examples, while more complex examples allow for a mix of automation and interactive tactics.
Raven~\citep{gupta2025-raven} is an SMT-based deductive verifier whose metatheory is based on a first-order fragment of Iris with support for logically atomic triples.

\section{Conclusion and Future Work}

This paper presented a completeness theorem for logically atomic specifications in Iris, showing that for any linearizable data structure, a corresponding logically atomic triple holds for the data structure's operations.
Although the proof is mechanized for \heaplang, the idea behind it should generalize to any Iris-based program logic that supports prophecy variables and satisfies Condition 18 of \citet{hostert2026-completeness}.
Besides showing that Iris's existing reasoning mechanisms are sufficient to verify atomic triples for any such data structure, this result also enables alternative linearizability proof techniques to be encoded into Iris as a means of deriving logically atomic triples.
We have carried out several such encodings in this paper, but it would be interesting to explore others, and to apply these results to verify additional data structures.

Although linearizability is an important correctness condition for concurrent data structures, there are many variations and generalizations beyond linearizability.
In future work, it would be interesting to attempt to generalize our completeness results to these variants of linearizability.
Another promising direction would be to use our completeness result as the starting point for a bridge to connect Iris to existing automated provers for linearizability.

\section*{Data Availability Statement}\addcontentsline{toc}{section}{Data Availability Statement}
The Rocq mechanization accompanying this work is available on GitHub at \url{https://github.com/u8cat/iris-logatom}.

\begin{acks}
  This work was supported in part by the \grantsponsor{NSF}{National Science Foundation}{} under Grant No.~\grantnum{NSF}{2319168} and Grant No.~\grantnum{NSF}{2524669}, as well as the \grantsponsor{Carlsberg Foundation}{Carlsberg Foundation}{} under Grant No.~\grantnum{Carlsberg Foundation}{CF23-0791}.
  Any opinions, findings, and conclusions or recommendations expressed in this material are those
  of the authors and do not necessarily reflect the views of these funding agencies.
\end{acks}

{
\interlinepenalty=10000
\bibliography{references}
}

\makeatletter
\par\bigskip\noindent{\small\normalfont\@received\par}
\renewcommand{\@received}{\@empty}
\makeatother

\newpage
\appendix
\crefalias{section}{appendix}
\crefalias{subsection}{appendix}
\crefalias{subsubsection}{appendix}

\section{Details about the Most General Client Construction}\label{sec:app:mgc}

This appendix gives additional details about the proofs involving the most general client (MGC) described in \Cref{sec:refine} and \Cref{sec:free-theorem}.

\paragraph{The Trace Library}
The MGC uses a trace library to record a trace at location $\loc_{\textit{tr}}$.
The library comes with two functions satisfying the specifications below.
\begin{mathpar}
  \inferhref{AhtEmit}{rule:wp-emit}
  {\traceInv{\loc_{\textit{tr}}}{I}}
  {\ahoareV
    {\All \phytr. \traceIs{\loc_{\textit{tr}}}{\phytr}\ast \pure{I(\phytr\dplus[(\tau,u)])}}
    {\function{emit}(\loc_{\textit{tr}},\tau,u)}
    {\Ret\TT.\traceIs{\loc_{\textit{tr}}}{\phytr\dplus[(\tau,u)]}}}
\and
\inferhref{AhtFresh}{rule:wp-fresh}
  {\traceInv{\loc_{\textit{tr}}}{I}}
  {\ahoareV[t]
  {\All \phytr.\traceIs{\loc_{\textit{tr}}}{\phytr}\ast\pure{\All \tau\notin\phytr.1\Ra I (\phytr\dplus[(\tau,u)])}}
  {\function{fresh}(\loc_{\textit{tr}},u)}
  {\Ret \tau.\traceIs{\loc_{\textit{tr}}}{\phytr\dplus[(\tau,u)]}\ast\pure{\tau\notin \phytr.1}}}
\end{mathpar}
Assertion $\traceIs{\loc_{\textit{tr}}}{\phytr}$ says the trace currently stored at $\loc_{\textit{tr}}$ is $\phytr$, and persistent assertion $\traceInv{\loc_{\textit{tr}}}{I}$ says the trace always satisfies a predicate $I\colon\Sigma^*\to\mProp$.
The user of the library is free to pick the event type $\Sigma$, provided each event is a pair of a tag $\tau$ and a value.
Function $\function{emit}$ appends a new event with tag $\tau$ to the trace, and $\function{fresh}$ (demonically) non-deterministically picks a fresh tag $\tau$ \wrt the existing trace, appends $(\tau,u)$ to the trace and returns the tag.
The two operations are logically atomic.

\paragraph{The Most General Client}
To construct the MGC, we first define an auxiliary function $\textit{wrapcr}$ to record the call-return trace
$$\textit{wrapcr}(\function{op},\textit{obj},x)\eqdef
\Let \tau := \function{fresh}(\loc_{\textit{tr}},\Inl(x)) in
\Let y := \function{op}(\textit{obj},x) in \function{emit}(\loc_{\textit{tr}},\tau,\Inr(y)); y$$
It is similar to \textit{wrapo} in \Cref{fig:wrapper}, but instead of resolving a prophecy variable, it records the trace using the trace library.
Here location $\loc_{\textit{tr}}$ is an external parameter that is allocated at meta-level.
Unlike \textit{wrapo}, \textit{wrapcr} takes the operation $\function{op}$ as an object-level parameter, so that the same wrapper can be applied to both sides of a refinement, as in the proof of \Cref{thm:lin-refines}.

To define $\textit{MGC}$, we first define a meta-level fixed point $\textit{concCalls}(\vv{\textit{req}}, \function{op},\textit{obj})\colon \Val^*\times\Expr\times\Expr\to\Expr$, and the $\textit{MGC}$ itself is a thin wrapper over \textit{{concCalls}}
\begin{align*}
\textit{concCalls}(\nil,\function{op},\textit{obj})\eqdef{}&!\loc_{\textit{tr}}\\
\textit{concCalls}(r::\vv{\textit{req}},\function{op},\textit{obj})\eqdef{}&\Fork{\textit{wrapcr}(\function{op},\textit{obj},r);\TT}; \textit{concCalls}(\vv{\textit{req}},\function{op},\textit{obj})\\
\textit{MGC}(\vv t)\eqdef{}&\Lam \function{op},\textit{obj}.\textit{concCalls}(\textit{extractRequests}(\vv t),\function{op},\textit{obj})
\end{align*}
The MGC spawns one thread for each operation in $\phytr$, where each thread calls $\textit{wrapcr}$ with the parameter extracted from $\phytr$, and the main thread of the MGC returns the trace stored at $\loc_{\textit{tr}}$, making the emitted events externally observable.

\paragraph{Supporting Lemmas}
The trace-agnostic specification $\tracespecname$ used in \Cref{thm:bdgjst-free} is defined as follows.
\begin{align}
&\mathleftalign{\tracespec[I]{\MgcInv}{\textit{vr}}{\function{init}}{\function{op}}\eqdef}\notag\\
&\hspace*{2em}\hoare{\TRUE}{\function{init}()}{\Ret \textit{obj}. \Exists \gamma.\All \loc_{\textit{tr}}.\traceIs{\loc_{\textit{tr}}}{\nil}\wand\pvs\MgcInv(\loc_{\textit{tr}},\gamma,\textit{obj})}\\
&\hspace*{2em}\pure{\textit{vr}(x)}\proves\hoare{\traceInv{\loc_{\textit{tr}}}{I}\ast\MgcInv(\loc_{\textit{tr}},\gamma,\textit{obj})}{\textit{wrapcr}(\function{op},\textit{obj},x)}{\Ret\_.\TRUE}\\
&\hspace*{2em}\persistent{\MgcInv}
\end{align}
Although the MGC is parameterized by a trace, its specification follows from this trace-agnostic condition.
\begin{lemma}\label{lem:mgc-spec}
If\/ $\tracespec[I]{\MgcInv}{\textit{vr}}{\function{init}}{\function{op}}$, then $\mgcspec[I]{\textit{vr}}{\phytr}{\function{init}}{\function{op}}$, where
\begin{align*}
&\mathleftalign{\mgcspec[I]{\textit{vr}}{\phytr}{\function{init}}{\function{op}}\eqdef}\notag\\
&\hspace*{2em}\pure{\All (\tau,\textlog{Call}(x))\in\phytr.\textit{vr}(x)}\proves\hoare{\traceInv{\loc_{\textit{tr}}}{I}\ast\traceIs{\loc_{\textit{tr}}}{\nil}}{\textit{MGC}(\phytr)(\function{op},\function{init}())}{\Ret \vv u.\pure{I(\vv u)}}
\end{align*}
\end{lemma}

The following lemma ensures that the MGC construction covers all possible adversarial steps.
\begin{lemma}[MGC Completeness]\label{lem:mgc-match-adv}
Given location $\loc_{\textit{tr}}$, if $(\nil,\nil,\pstate)\advsteps[\textit{obj}.\function{op}/\textit{vr}]{\phytr} (\vv\pi',\vv\expr',\pstate')$ and $\loc_{\textit{tr}}\notin \dom(\pstate')$, then
$([\textit{MGC}(\phytr)(\function{op},\textit{obj})],\textit{compose}(\pstate,\nil))\tpsteps(\phytr::\ltrans\vv\pi'\rtrans\dplus\vv\expr',\textit{compose}(\pstate',\phytr))$,
where $\textit{compose}(\pstate,\phytr)\eqdef\mapinsert{\loc_{\textit{tr}}}{\phytr}{\pstate}$ and $\ltrans\vv\pi\rtrans$ is a function to translate threads in the $\advsteps{}$ relation to the threads spawned by the MGC.\footnote{This lemma statement is simplified to omit technical details about the trace library, \eg the actual trace library stores a mutex at $\loc_{\textit{tr}}+1$ and spawns a child thread to generate non-deterministic fresh tags.
}
\end{lemma}

Symmetrically, every value returned by the MGC must be a trace generated by some adversarial steps.
\begin{lemma}[MGC Soundness]\label{lem:adv-match-mgc}
Given location $\loc_{\textit{tr}}$, if $([\textit{MGC}(\vv*t0)(\function{op},\function{init}())],\textit{compose}(\pstate_0,\nil))\tpsteps(\val::\_,\_)$ ($\val\in\Val$), $\loc_{\textit{tr}}\notin\dom(\pstate_0)$, and the technical conditions below hold, then $\val\in\Sigma^*$ and there exists $\textit{obj}\in\Val$ and $\rho'$ such that
$([\function{init}()],\pstate_0)\tpsteps([\textit{obj}],\pstate)$ and $(\nil,\nil,\pstate)\advsteps[\textit{obj}.\function{op}/\textit{vr}]{\val}\rho'$.
The technical conditions are
\begin{itemize}
\item $\All \pstate.\safe{\function{init}()}{\pstate}\land\nofork{\function{init}()}{\pstate}$,
\item $\All \pstate_0,\textit{obj},\pstate.([\function{init}()],\pstate_0)\tpsteps([\textit{obj}],\pstate)\Ra \cfgsafe{\nil}{\nil}{\pstate}$, and
\item the parameter of every call event in $\vv*t0$ satisfies $\textit{vr}$.
\end{itemize}
\end{lemma}

Finally, we turn to prove \Cref{thm:bdgjst-free}.
\begin{proof}[Proof of \Cref{thm:bdgjst-free}]
Conclusion (1) is proved by the soundness direction of \Cref{thm:sound-complete}.
For conclusion (2), assume $([\function{init}()],\varnothing)\tpsteps([\textit{obj}],\pstate)$ and $(\nil,\nil,\pstate)\advsteps[\textit{obj}.\function{op}/\textit{vr}]{\phytr}(\vv{\pi}, \vv{e}, \pstate')$.
By \Cref{lem:mgc-match-adv}, there exists an execution such that $\textit{MGC}(\phytr)(\function{op},\textit{obj})$ returns $\phytr$.
Further, by \Cref{lem:mgc-spec} and the soundness direction of \Cref{thm:sound-complete}, every result of $\textit{MGC}(\phytr)(\function{op},\textit{obj})$ satisfies $I$.
Therefore, $I(\phytr)$.

For the non-stuckness, also because $\textit{MGC}(\phytr)(\function{op},\textit{obj})$ is safe, every expression in $\ltrans\vv\pi\rtrans\dplus\vv\expr$ is not stuck.
This concludes the non-stuckness of $\vv\expr$, and for $\ltrans\vv\pi\rtrans$, it is not hard to show that the reverse translation preserves non-stuckness.
\end{proof}

\end{document}